\documentclass{article}[12pt]
\usepackage[utf8]{inputenc}

\usepackage{amsmath}
\usepackage{amssymb}
\usepackage{amsthm}
\usepackage{graphicx}
\usepackage{subcaption}
\usepackage{listings}
\usepackage{subcaption}
\usepackage[version=4]{mhchem}
\usepackage{float}

\usepackage{geometry}
\usepackage{bbm}
\usepackage[dvipsnames]{xcolor}
\usepackage{footnote}
\usepackage{comment}

\usepackage{footnote}
\usepackage{ulem}
\usepackage{graphicx}
\usepackage{subcaption}
\usepackage{listings}
\usepackage{hyperref}
\usepackage{animate}
\usepackage{comment}
\usepackage{geometry}
\usepackage{bbm}
\usepackage{babel}
\usepackage{authblk}
\newtheorem{proposition}{Proposition}
\newtheorem{assumption}{Assumptions}
\newcommand{\specificthanks}[1]{§}

\title{Trait-structured chemotaxis: Exploring ligand-receptor dynamics and travelling wave properties in a Keller-Segel model}
\author{Viktoria Freingruber  
 \thanks{Delft Institute of Applied Mathematics, Delft University of Technology, Mekelweg 4, Delft, 2628 CD, The Netherlands (v.e.freingruber@tudelft.nl).} \textsuperscript{, \specificthanks{4}}\, Tommaso Lorenzi \thanks{Department of Mathematical Sciences “G. L. Lagrange”, Politecnico di Torino, Corso Duca degli Abruzzi, 24, Torino, 10129, Italy
 (tommaso.lorenzi@polito.it).}  \, Kevin J. Painter \thanks{Dipartimento Interateneo di Scienze, Progetto e Politiche del Territorio, Politecnico di Torino, Viale Pier Andrea Mattioli, 39, Torino, 10125, Italy (kevin.painter@polito.it).} \, Mariya Ptashnyk \thanks{Department of Mathematics, The Maxwell Institute for Mathematical Sciences, Heriot-Watt University, Edinburgh EH14 4AS, Scotland, UK (m.ptashnyk@hw.ac.uk).}}
\date{}

\begin{document}

\maketitle

\begin{abstract}
A novel trait-structured Keller-Segel model that explores the dynamics of a migrating cell population guided by chemotaxis in response to an external ligand concentration is derived and analysed. Unlike traditional Keller-Segel models, this framework introduces an explicit representation of ligand-receptor bindings on the cell membrane, where the percentage of occupied receptors constitutes the trait that influences cellular phenotype. The model posits that the cell's phenotypic state directly modulates its capacity for chemotaxis and proliferation, governed by a trade-off due to a finite energy budget: cells highly proficient in chemotaxis exhibit lower proliferation rates, while more proliferative cells show diminished chemotactic abilities. The model is  derived from the principles of a biased random walk, resulting in a system of two non-local partial differential equations, describing the densities of both cells and ligands. Using a Hopf-Cole transformation, we derive an equation that characterises the distribution of cellular traits within travelling wave solutions for the total cell density, allowing us to uncover the monotonicity properties of these waves. Numerical investigations are conducted to examine the model’s behaviour across various biological scenarios, providing insights into the complex interplay between chemotaxis, proliferation, and phenotypic diversity in migrating cell populations.

\end{abstract}


\section{Introduction}

Chemotaxis refers to the guidance of a cell or organism via a response to the gradient of a chemical concentration, and has been studied intensively in bacteria, cells, and animals through a combination of experimental and theoretical approaches. Following its introduction in the early 1970s, the Keller-Segel model has become the preeminent modelling tool for representing chemotaxis phenomena at a macroscopic level. The Keller-Segel model was inspired by observations of macroscopic organisation within two biological systems: the self-organisation of dispersed amoebae during \textit{Dictyostelium discoideum} mound formation \cite{bonner2015cellular, keller1970initiation}, and the formation of travelling bands within \textit{E.coli} populations \cite{adler1966chemotaxis,keller1971traveling}. The essential features of these two phenomena can be captured within a system of partial differential equations (PDE) for a (cell) population density, $p(x,t)$, that responds to a chemoattractant, $c(x,t)$, of the form
\begin{equation} \label{eq:originalKS}
\begin{split}
    \partial_t p &= \nabla \cdot \left[ D(p,c) \nabla p - A (p,c) p \nabla c \right] + k_1 (p,c), \\
    \partial_t c &= d \Delta c + k_2(p,c),
\end{split}
\end{equation}
for $x \in \mathbb{R}^n, t\geq 0$, and with appropriate initial conditions. In \eqref{eq:originalKS} $D(p,c)$ is the diffusion coefficient of the cells,  $A (p,c)$ is the chemotactic sensitivity of cells to the attractant gradient, $k_1(p,c)$ describes population growth and death, $k_2(p,c)$ describes attractant kinetics, and $d$ is the attractant diffusion coefficient. The above equations, along with their numerous variations, have been subject to intense research, both in terms of their mathematical properties, e.g. see \cite{arumugam2021keller,bellomo2015toward, hillen2009user, horstmann20031970I, horstmann20031970II}, and their application to describe the organisation of different populations \cite{painter2019mathematical}. 

The specific form of system~\eqref{eq:originalKS} is the result of a number of simplifying assumptions, but the extent to which these remain reasonable will vary from process to process. Chemotaxis responses in cells often involve a mechanism in which a freely diffusing extracellular ligand (the chemoattractant) binds to a cell surface receptor, thereby modulating the internal signalling pathways that control motility. The simplest assumption is to ignore these complexities: The response of cells is according to the local ligand concentration gradient, whereas the ligand kinetics are based on its production and/or decay, but any binding of ligand to cell surface receptors is omitted. A slightly more involved approach is to account for ligand-receptor binding within the chemotactic response -- by assuming that the chemotactic response depends on differences in receptor occupancy – but subsequently assume fast binding and unbinding (relative to other timescales) so that a quasi-steady-state approximation can be applied. As a result, the chemotactic sensitivity functions depend on the concentration of the ligand (e.g. the so-called receptor law), but again any receptor binding consequences are often excluded from chemical kinetics \cite{ hillen2009user, segel1977theoretical}.

What if receptor binding and unbinding occur on slower timescales? While receptor-ligand residence times can be tricky to determine, at a general level these can be highly variable and vary from an order of seconds to an order of hours \cite{tummino2008residence}. If longer residence times are the case, there are several potential consequences. First, if the extracellular ligand is available at low concentration levels, then the level of freely available ligand could become significantly reduced as it becomes increasingly sequestered to the cell surface. Second, since cells are migrating, the ligand attached to the cell surface will be transported, potentially released at a point distant from where it was bound. Third, the cell population will naturally become structured according to the amount of ligand attached to its surface. As many chemoattractants also serve other functions, for example, acting as growth factors or nutrients, the amount of bound ligand may also regulate the phenotype of a cell, e.g. triggering both motility and proliferative responses.

Previous research on trait-structured models where the evolution equation for the population density is described by non-local diffusion-(advection-)reaction equations includes, among others, \cite{arnold2012existence, benichou2012front, berestycki2015existence, bouin2014travelling, bouin2012invasion, bouin2017super, domschke2017structured, engwer2017structured, hodgkinson2018computational, lorenzi2022invasion, macfarlane2022individual, turanova2015model}, for a review see \cite{lorenzi2024phenotype}. In particular, work on trait-structured Keller-Segel-type models has been undertaken in \cite{lorenzi2022trade, mattingly2022collective}. Our study is different, as the traits are directly linked to membrane occupancy, which follows from biochemical membrane attachment processes.

In this paper, we develop a comprehensive framework that explicitly integrates receptor-ligand binding dynamics into the modelling of cellular structuring, responses, and ligand kinetics. Our approach involves deriving evolution equations for a continuously trait-structured cell population, along with the concentration of a freely diffusing ligand that functions as both a chemoattractant and, potentially, a nutrient. The trait in our model is defined as the percentage of membrane receptors occupied by the ligand, with the ``membrane occupancy" directly influencing the cell's phenotype and determining its capacity for division or chemotaxis. Cells dynamically alter their trait by binding or releasing ligand molecules to or from their receptors, thereby sequestering ligands from the environment and redistributing them as they migrate. We formulate an individual-based on-lattice branching biased random walk model, which we then coarse-grain to derive a corresponding system of coupled non-local PDEs of the Keller-Segel type.

The classical Keller-Segel model has been extensively studied in the context of invading waves. However, our focus is to illuminate the structuring of a heterogeneous cell population within such waves and to explore how intra-population variability influences the properties of these travelling waves. Addressing these questions could have significant biological implications, particularly in contexts such as metastatic cancer. For example, does a more migratory cell type invariably lead at the forefront of an invading wave? Moreover, could a few migratory cells lag behind, subsequently proliferating and, when re-exposed to an attractant, initiate a new wave of invasion?

The paper is organised as follows: In Section~\ref{sec:Micromodel} we introduce an individual-based on-lattice biased random walk model in the space-trait domain and formulate assumptions on the mobility and reaction terms. In Section~\ref{sec:Macromodel} the corresponding PDE model is presented after it is derived using formal asymptotic techniques. The main results of formal travelling wave analysis are stated in Section~\ref{sec:TW}, with numerical results following in Section~\ref{sec:numres}. We summarise and discuss our results in Section~\ref{sec:discussion}.

\begin{figure}
\centering
\begin{subfigure}{0.4\textwidth}
    \includegraphics[width=\textwidth]{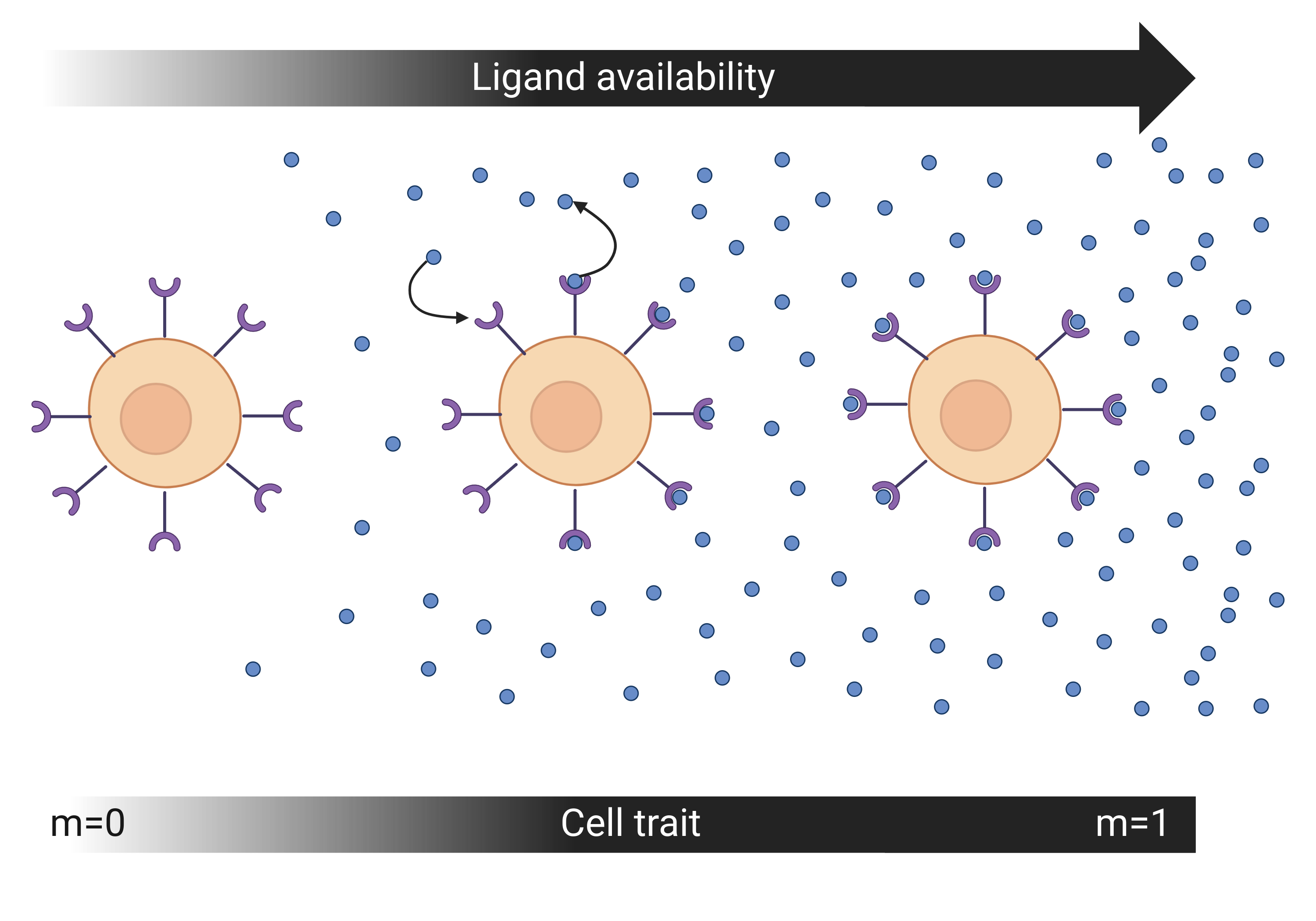}
    \caption{}
    \label{fig:celltrait}
\end{subfigure}
\hfill
\begin{subfigure}{0.5\textwidth}
    \includegraphics[width=\textwidth]{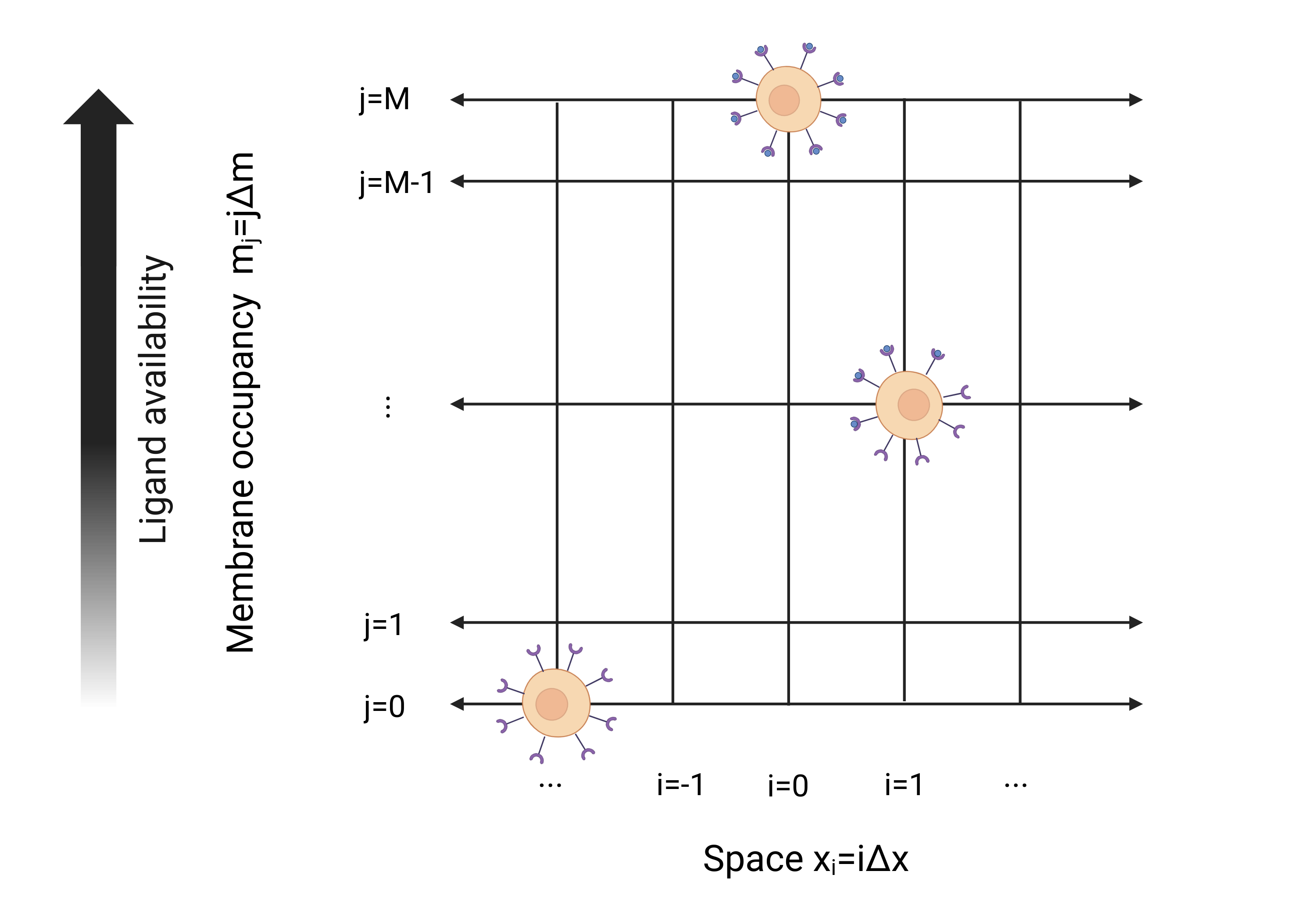}
    \caption{}
\end{subfigure}
\caption{(a) Cells change their trait by attaching or detaching ligands to or from their membrane receptors; (b) Biased random walk of cells as point-particles in space (displayed  on the horizontal axis) and trait (displayed  on the vertical axis)}
\label{Ch2:Fig:schematic_randomwalk}
\end{figure}

\begin{figure}
\centering

    \includegraphics[width=0.55\textwidth]{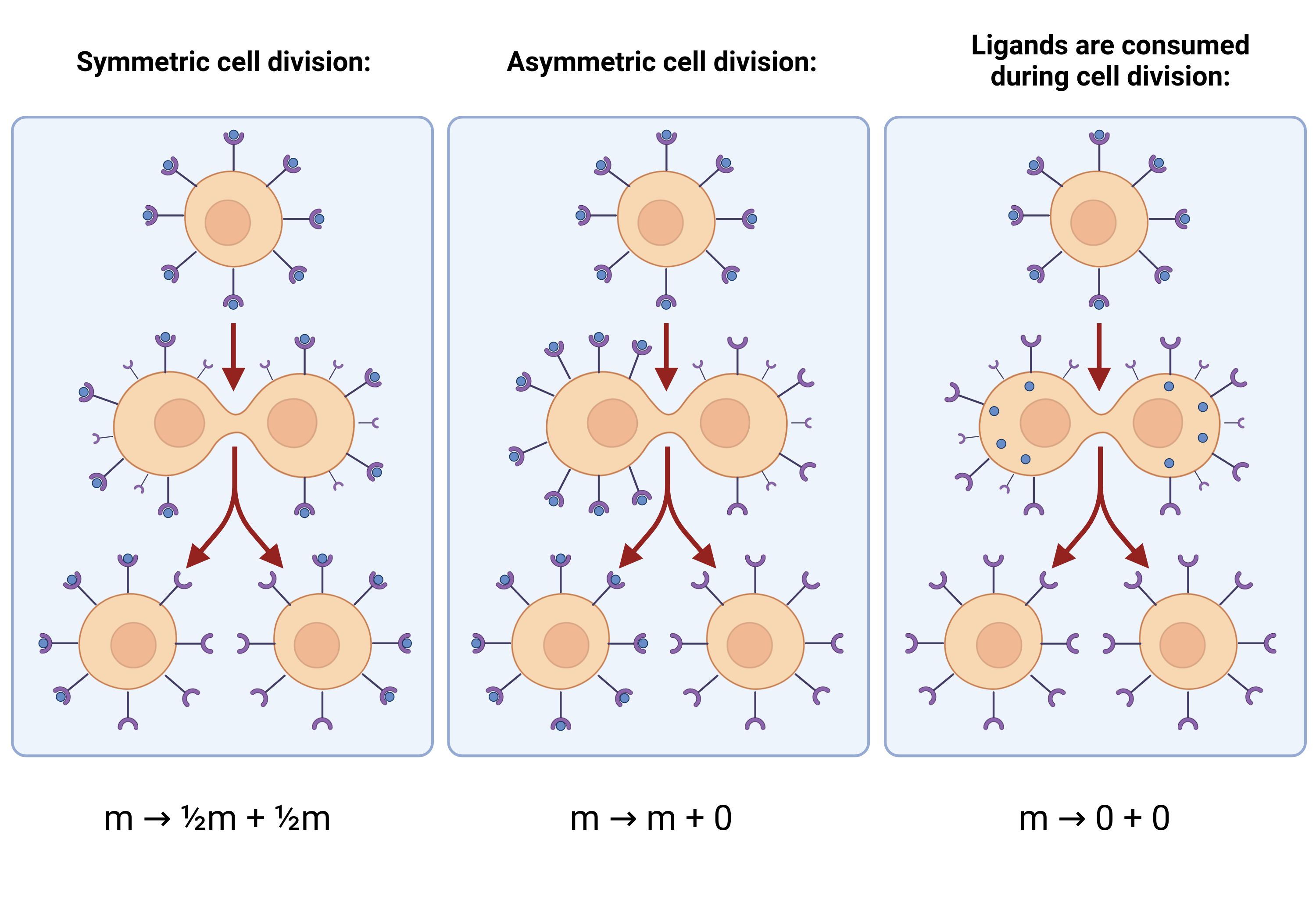}

\caption{ Different division processes of cells in terms of trait inheritance: Symmetric division, Asymmetric division, Consumptive division}
\label{fig:schematic_division}
\end{figure}

\section{Individual-based on-lattice model}
\label{sec:Micromodel}

Numerous derivations of continuous Keller-Segel-type models have been made from an underlying random walk or an individual-based model, including a position jump random walk on a lattice \cite{stevens1997aggregation}, kinetic transport equations for a velocity jump process \cite{othmer2002diffusion}, and a stochastic interaction particle system \cite{stevens2000derivation}. In the present work, our intention is to explicitly account for receptor-ligand binding processes, such that a (chemotactic) cell population becomes structured according to its level of receptor occupancy. To this end, we adopt a branching position-jump random walk model related to those considered within a number of recent studies \cite{ chisholm2016evolutionary,lorenzi2023derivation, macfarlane2022individual}. Specifically, we consider an individual-based on-lattice random walk model, in which single (migrating and dividing) cells move between adjacent sites on a 2D lattice, where one dimension represents position in space and the other represents membrane occupancy, see Figure~\ref{Ch2:Fig:schematic_randomwalk}. We distinguish between free and bound ligands: free ligands diffuse through space, but can become bound to the surface of cells and are hence removed from the cohort of free ligands. Vice versa, bound ligands can detach from the cell membrane and become free ligands again. 

First, we discretise each of time, space and membrane occupancy state as
\begin{equation*}
\begin{split}
 t_k &= k \Delta t \in \mathbb{R}^+, \, \text{ for } \;  k \in \mathbb{N}_0, \; \Delta t \in \mathbb{R}^+,  \\
 x_i &= i \Delta x \in \mathbb{R}, \, i \in \mathbb{Z}, \; \Delta x \in \mathbb{R}^+, \\
 m_j &= j \Delta m \in [0,1], \, 0 \leq j \leq M, M \in \mathbb{N}, \; \Delta m = 1/M.
 \end{split}
\end{equation*} 
Here, $\Delta t, \Delta x, \Delta m$ are the discretisation steps for time, space, and membrane occupancy, respectively.

Cells at position $x_i$ with membrane occupancy $m_j$ at time $t_k$ are regarded as point-particles located at the lattice coordinate $(x_i,m_j)$ and at time $t_k$. We denote the number of cells at $(x_i,m_j)$ at time $t_k$ as $N_k^{i,j}$. We also assume that a certain amount of free (i.e. unbound) ligand can occupy the small region around lattice point $x_i$ at time $t_k$, and denote by $C^i_k$ the number of free ligand molecules associated with lattice site $x_i$ at time $t_k$.

The cell population density and corresponding total cell density will be denoted as
\begin{equation*}
p(x_i,t_k,m_j) := \frac{N_k^{ij}}{\Delta x \Delta m}, \quad P(x_i,t_k) := \Delta m \sum_{j=0}^M p(x_i,t_k,m_j).
\end{equation*}
Similarly, the free ligand density (or concentration)  will be denoted as
\begin{equation*}
c(x_i,t_k) := \frac{C^i_k}{\Delta x}.
\end{equation*}
Where convenient, we will further abbreviate these expressions as
\begin{equation*}
    p^{i,j}_k := p(x_i,t_k,m_j), \quad c^i_k := c(x_i,t_k), \quad P^i_k := P(x_i,t_k).
\end{equation*}

\textbf{Movement of cells.}
For the movement of cells through space, we assume a (positive) chemotactic-type movement: a biased movement towards adjacent positions with higher levels of free ligand. Specifically, we assume that the probabilities of changing position depend on a term representing unbiased movement $D(m_j)$ and a term representing directed migration $A(m_j) (c^{i\pm 1}_k - c^i_k)$. The latter term depends on the difference in ligand concentration between the target position and current position, and follows a standard approach to incorporate chemotactic behaviour in a simple and intuitive way \cite{stevens1997aggregation}. Note that the coefficients $D(m_j)$ and $A(m_j)$ are allowed to depend on the level of membrane occupancy, which allows the degree of ligand binding to directly impact cell movement. The probability of a cell with trait $m_j$ moving from location $x_i$ to $x_{i\pm 1}$ between time $t_k$ and $t_{k+1}$ is defined as
\begin{equation*}
\mathcal{P}^{i \to i \pm 1, j}_k = \frac{\Delta t }{\Delta x ^2} D(m_j) +\frac{\Delta t }{2 \Delta x ^2} A(m_j) \left( c^{i \pm 1}_k- c^i_k  \right)_+.
\end{equation*}
The probability of not moving between $t_k$ and $t_{k+1}$ is then
\begin{equation*}
\mathcal{P}^{i\to i,j}_k = 1 - \mathcal{P}^{i\to i+1,j}_k - \mathcal{P}^{i\to i-1,j}_k.
\end{equation*}
Note that the time step $\Delta t$ is assumed to be sufficiently small such that all probabilities are between $0$ and $1$.

\textbf{Attachment/ detachment of ligands.}
In our model we consider the percentage of occupied membrane receptors as an independent variable that can take values in $[0,1]$, where the boundaries of the interval represent an unoccupied and a fully occupied membrane, respectively. Note that as a simplification we will assume that there is a fixed  constant total number of receptors at the surface of each cell, so that the only way receptor occupancy changes is through ligand binding or unbinding from receptors. 
We assume that the binding of the ligand [c] to the receptor [R] follows the simple reversible reaction \ce{ [c] + [R]  <=>[K_1][K_2] [cR] }, where we denote by $K_1$ and $K_2$ the rate of attachment and detachment, respectively. 

Applying a Law of Mass Action, the probability that a cell in state $m_j$ switches to state $m_{j+1}$ between time $t_k$ and $t_{k+1}$ by attaching ligand to its membrane receptors is
\begin{equation}
\label{eq:Pattach}
\mathcal{P}^{i,j \to j+1}_{k,\text{attach}} =  \frac{\Delta t}{\Delta m } K_1 c^i_k \omega_1(m_j),
\end{equation}
where $\omega_1(m_j) \in C([0,1])$ describes the ability of a cell with occupancy $m_j$ to attach ligands. In particular, to indicate that cells with a fully occupied membrane cannot attach ligand, we have $\omega_1(1)=0$. By similar reasoning, the probability of changing from state $m_j$ to state $m_{j-1}$ between time $t_k$ and $t_{k+1}$ through ligand detachment is 
\begin{equation}
\label{eq:Pdetach}
\mathcal{P}^{j \to j-1}_{\text{detach}} = 
\frac{\Delta t}{\Delta m } K_2 \omega_2(m_j),
\end{equation}
where $\omega_2(m_j) \in C([0,1])$ represents the ability of a cell with occupancy $m_j$ to release bound ligand into the environment. We set $\omega_2(0) = 0 $, as a cell with an unoccupied membrane receptor cannot release ligand. As a consequence of \eqref{eq:Pattach} and \eqref{eq:Pdetach}, the probability of not changing membrane occupancy between times $t_k$ and $t_{k+1}$ is
\begin{equation*} \mathcal{P}^{j \to j }_k = 1 - \mathcal{P}^{j\to j+1}_{k,\text{attach}} - \mathcal{P}^{j\to j-1}_{\text{detach}}.
\end{equation*}

\textbf{Cell proliferation.}
We assume that a cell can divide into two daughter cells with some probability. In doing so, we must consider what happens to any ligand already bound to a dividing cell and, in particular, ensure ligand is not artificially created during this process. The precise rule for how bound ligand is shared between daughter cells will vary according to various factors, including the scope of ligand action on cells. We make the following assumptions. Let $B(m_j,m_h)$ denote the probability that a cell with membrane occupancy $m_j$ creates an offspring with membrane occupancy $m_h$. For $B(m_j,m_h)$ we use evaluations at $(m_j,m_h)$ of a (bimodal) probability density function. For a mother cell with membrane occupancy $m_j$, the average occupancies of the two daughter cells must satisfy $2\sum_{h \leq M} B(m_j,m_h)m_h \leq m_j$ to ensure ligand is not created during division. 
We will consider three division strategies, where a cell of trait $m_j$ divides into:
\begin{itemize}
    \item[(i)] two cells of trait $m_j/2$ on average (\textit{symmetric division}, Figures \ref{fig:schematic_division}, \ref{fig:schematic_division2}a);
    \item[(ii)] one cell of trait $m_j$ and one of trait $0$ on average (\textit{asymmetric division}, Figures \ref{fig:schematic_division}, \ref{fig:schematic_division2}b);
    \item[(iii)] two cells of trait $0$ on average (\textit{consumptive division}, Figures \ref{fig:schematic_division}, \ref{fig:schematic_division2}c).
\end{itemize} 
For the rate of division taking place $b(m_j,c^i_k)$ we consider two cases: 
\begin{itemize}
    \item[(i)] cells do not rely on the presence of free ligand to divide (i.e. $b\equiv b(m_j)$),
    \item[(ii)] cells rely on the presence of ligand to divide, that is, the ligand acts as a nutrient (i.e. $b \equiv b(m_j,c^i_k)$).
\end{itemize}
The probability of a cell with membrane occupancy $m_j$ dividing into a cell with membrane occupancy~$m_h$ is then \begin{equation*}\mathcal{P}^{i, j \to h }_{k, \text{div} }  = b(m_j,c^i_k) B(m_j,m_h)  \Delta t.
\end{equation*}
Cells may die with a certain probability that is independent of their trait and depends on the total cell density due to saturating effects on population growth. In that process, bound ligands do not re-enter the system and decay together with the cell. The probability of dying between times $t_k$ and $t_{k+1}$ is 
\begin{equation*}
    \mathcal{P}^{i,j}_{k,\text{die}} = \gamma (P^i_k) \Delta t,
\end{equation*}
with $ \gamma \in C^1(\mathbb{R}^+)$. The probability of a cell of trait $m_j$ staying quiescent is then 
\begin{equation*}
\mathcal{P}^{i,j}_{k,\text{quiescent}} = 1 - \mathcal{P}^{i,j}_{k,\text{die}} - \mathcal{P}^{i,j \to h}_{k, \text{div}}.
\end{equation*}

\textbf{Diffusion of ligands.}
Free ligand is assumed to move through space in a purely random manner, that is, the probability of moving to an adjacent site within a time interval $\Delta t$ is constant and equal for all directions, i.e.
\begin{equation*}
\mathcal{P}_{\text{move}}^c=\frac{\Delta t}{\Delta x ^2} d,
\end{equation*}
with some constant $d>0$.

\textbf{Ligand attachment to and detachment from membranes.}
Attachment of free ligand to the cell membrane removes it from the extracellular environment. This ligand can subsequently be released later in time. The probabilities of free ligand being removed from and released correspond to those in \eqref{eq:Pattach}, \eqref{eq:Pdetach}. We assume that during one attachment/ detachment event, a number of $\sigma \Delta m  $ ligands are removed/ released, where $\sigma$ denotes the number of attached ligands required to occupy a fraction of $\Delta m$ of a cell's membrane receptors.

\textbf{Ligand consumption by cells.}
We assume that ligands are consumed by cells with probability
\begin{equation*}\mathcal{P}_{\text{consumed}} = \Delta t Q(P^i_k),
\end{equation*}
where $Q  \in C^1(\mathbb{R}^+)$.
Note that the ligand is consumed, regardless of whether the ligand also acts as a nutrient required for cell division.\\

\textbf{Master equations.} The corresponding master equation for the population density of cells at location $x_i$ at time $t_k$ with membrane occupancy state $m_j$ reads as follows

\begin{equation}
\label{eq:masterP}
\begin{split}
    p^{i,j}_{k+1} - p^{i,j}_k= & \underbrace{\mathcal{P}^{i-1 \to i,j}_k p^{i-1,j}_k +\mathcal{P}^{i+1 \to i,j }_k p^{i+1,j}_k - \left( \mathcal{P}^{i \to i-1,j}_k + \mathcal{P}^{i \to i+1,j }_k \right) p^{i,j}_k }_{\text{biased random walk in space}} \\ 
    & +\underbrace{\mathcal{P}^{i,j-1\to j}_{k,\text{attach}} p^{i,j-1}_k +\mathcal{P}^{j+1 \to j}_{\text{detach}} p^{i,j+1}_k - \left( \mathcal{P}^{i,j \to j+1}_{\text{attach}} + \mathcal{P}^{j \to j-1}_{\text{detach}} \right) p^{i,j}_k }_{\text{change of membrane occupancy}} \\ 
    & - \underbrace{\mathcal{P}^{i,j}_{k,\text{die}} p^{i,j}_k}_{\text{probability of dying}} - \underbrace{\frac{1}{2}   \Delta m \sum_{h=0}^M \mathcal{P}^{i,j \to h}_{k,\text{div}}  p^{i,j}_k}_{\text{removal by division}} + \underbrace{\Delta m \sum_{h=0}^M \mathcal{P}^{i,h\to j}_{k,\text{div}} p^{i,h}_k}_{\text{creation by division}},
\end{split}
\end{equation}
while the master equation for the ligands is
\begin{equation}
\label{eq:masterC}
\begin{split}
    c^i_{k+1} - c^i_k = & \underbrace{\mathcal{P}_{\text{move}}^c \left( c^{i-1}_k + c^{i+1}_k - 2 c^i_k \right) }_{\text{random movement}}  - \underbrace{\mathcal{P}_\text{consumed} c^i_k}_{\text{consumption}}\\
    &\underbrace{- \Delta m  \sigma  \Delta m \sum_{h=0}^{M} \mathcal{P}_\text{attach}^{i,h \to h + 1} p^{i,h}_k+ \Delta m  \sigma  \Delta m \sum_{h=0}^{M} \mathcal{P}_\text{detach}^{h \to h -1} p^{i,h}_k}_{\text{attachment/ detachment}},
\end{split}
\end{equation}
where $i \in \mathbb{Z}$, $0 \leq j \leq M$, $k \in \mathbb{N}$.
\section{The PDE model}
\label{sec:Macromodel}

We formally derive the corresponding PDE model from the discrete model \eqref{eq:masterP}, \eqref{eq:masterC} by following the standard approach of expanding about $(x,t,m)$ and taking $\Delta x, \Delta t, \Delta m \to 0$ in an appropriate way, see e.g.~\cite{stevens1997aggregation}. We use the following notation $p(x,t,m) := p^{i,j}_k = p(x^i,t^k,m^j), \quad c(x,t) := c^i_k = c(x^i,t^k)$. If the functions $p$ and $c$ are three times differentiable w.r.t.~$x$ and twice w.r.t. $t$, and $p$ is differentiable twice w.r.t.~$m$ and $t$, respectively, and with $\Delta x, \Delta m, \Delta t$ small enough, we can apply the Taylor expansions to  \eqref{eq:masterP} and \eqref{eq:masterC} about $(x,t,m)$. Inserting all the probabilities defined in Section~\ref{sec:Micromodel} and letting  $\Delta t, \Delta x, \Delta m \to 0$ results in the following non-local diffusion-advection-reaction and diffusion-reaction equations

\begin{equation}
\label{eq:macromodel}
\begin{aligned}
& \frac{\partial p}{\partial t}   =  \frac{\partial }{\partial x } \left[ D(m) \frac{\partial p}{\partial x }  - A(m) \, p \, \frac{\partial c}{\partial x }  \right] 
  + \frac{\partial}{\partial m} \Big[ p \big( K_2 \omega_2(m)  - K_1 \omega_1(m) c \big) \Big]  \\
  & \qquad - \gamma (P)\, p 
+ \int_0^1\!\!\! b(y,c) B(y,m) p(x,t,y) dy - \frac{1}{2} p \, b(m,c) \int_0^1\!\!\! B (m,y) dy, \\ 
& \frac{\partial c}{\partial t}   =  d \frac{\partial^2 c}{\partial x^2}  - Q(P) c
 +\sigma \Big[K_2 \int_0^1\!\!\! \omega_2(y) p(\cdot, \cdot,y) dy -  K_1 c \int_0^1\!\!\! \omega_1(y) p(\cdot, \cdot, y) dy\Big], \\ 
& p(x,0,m) =  p_0(x,m)\;\;\; \text{ in } \mathbb{R} \times (0,1), \qquad c(x,0 ) = c_0(x)  \;\;\; \text{ in }\mathbb{R}, \\
& 0=p(x,t,m) \big(\omega_2 (m) K_2 - \omega_1 (m) K_1 c(x,t) \big)   \quad \text{at }  m=0,1,\, (x,t) \in \mathbb{R} \times \mathbb{R}^+,\\
& P(x,t) :=  \int_0^1 \!\!\!p(x,t,m) \, dm, 
\end{aligned}
\end{equation}
where $K_1, K_2, d, \sigma$ are non-negative constants and $D(m), A_1(m), \gamma(P),Q(P), b(m,c)$, $B(m,y)$ are non-negative, sufficiently regular functions. The reader is referred to Section~\ref{sec:Micromodel} for more details about these functions.

The equation for the cell population density in \eqref{eq:macromodel} has the form of a trait-structured Keller-Segel model for chemotaxis with a growth-fragmentation process that describes a potentially asymmetric division process of cells. The following mechanisms naturally appear in this continuous model from the assumptions made on the individual cell scale: Movement of cells as a result of diffusion and chemotaxis, attachment and detachment of free ligands to and from the cell membrane, decay of ligands, and proliferation of cells.

\section{Travelling wave solutions}
\label{sec:TW}

The results presented here are based on techniques employed in ~\cite{ lorenzi2023derivation,lorenzi2022trade, lorenzi2022invasion}, which build on methods developed in~\cite{barles2009concentration, diekmann2005dynamics, lorz2011dirac, perthame2006transport, perthame2008dirac}. We assume diffusion happens on a slower timescale and cell proliferation on a faster timescale. By introducing a small parameter $\epsilon$, we obtain the following rescaled system in $\mathbb{R} \times \mathbb{R}^+ \times (0,1)$:
\begin{equation} \label{eq:system}
\begin{aligned}
&\frac{\partial p_\epsilon}{\partial t}  =  \frac{\partial }{\partial x}  \left[ \epsilon \tilde{D} (m) \frac{\partial p_\epsilon}{\partial x}    - A(m) \, p_\epsilon  \frac{\partial c_\epsilon }{\partial x}   \right]   + \frac{\partial}{\partial m} \Big[ p_\epsilon  \big( K_2 \omega_2(m) - K_1 \omega_1(m) c_\epsilon  \big) \Big] \\
& \qquad - \frac{1}{\epsilon } \tilde{\gamma} (P_\epsilon) p_\epsilon + \frac{1}{\epsilon} \Big[\int_0^1\!\!\! \tilde{b}(y,c)B (y,m) p_\epsilon (\cdot,\cdot,y) dy  - \frac{1}{2} p_\epsilon \, \tilde{b} (m,c)\int_0^1\!\!\! B(m,y) dy \Big], \\
&\frac{\partial c_\epsilon }{\partial t}  =  \epsilon \tilde{d}\, \frac{\partial^2 c_\epsilon}{\partial x^2}    -  Q(P_\epsilon)c_\epsilon  + \sigma \Big[K_2 \int_0^1\!\!\! \omega_2(y) p_\epsilon (\cdot,\cdot,y) dy  -  K_1 c_\epsilon  \int_0^1\!\!\! \omega_1(y) p_\epsilon(\cdot,\cdot,y) dy \Big], \\
& p_\epsilon (x,0,m) =  p_0(x,m) \qquad  \text{in } \mathbb{R} \times (0,1), \quad \qquad 
 c_\epsilon (x,0 ) = c_0(x) \qquad  \text{in } \mathbb{R},\\
& 0= p_\epsilon (x,t,m) (\omega_2 (m) K_2 - \omega_1 (m) K_1 c_\epsilon (x,t) )  \quad \text{at} \quad m=0,1, \, (x,t) \in \mathbb{R} \times \mathbb{R}^+ \\
& P_\epsilon (x,t) :=  \int_0^1 p_\epsilon (x,t,m) \, dm,
\end{aligned}
\end{equation}
where $\tilde{D}(m) = D(m)/\epsilon$, $\tilde{\gamma}(P) = \epsilon \gamma(P)$, $\tilde{b}(m,c) = \epsilon b(m,c)$, and $\tilde{d} = d/\epsilon$.
To formally derive the properties of travelling wave solutions for the re-scale system \eqref{eq:system} and to determine the minimum wave speed we consider the following assumptions. For readability, we omit the tilde symbols on the re-scaled coefficients in the remainder of the paper.
\begin{assumption}
\label{assum}
\begin{itemize}
\item[(A0)] The consumptions and decay terms are
\begin{equation*} Q(P) := qP, \quad \gamma(P) := \beta P, \quad \text{with } \; q,\beta > 0.
\end{equation*}
\item[(A1)] The initial conditions for the ligand $c_0 \in C(\mathbb R)\cap L^2(\mathbb R)$ satisfy
\begin{equation}
\label{eq:ass:c0}
     c_0 > \frac{\sigma K_2}{q}.
\end{equation}
\item[(A2)] The sensitivity function $A \in C^1([0,1])$ is non-decreasing.
\item[(A3)] The division rate $b$ is either,         \textbf{(I)} independent of the chemoattractant, i.e. $b(m) \equiv b$ and we assume 
\begin{equation*}
\frac{d}{dm }b(m) <0, \quad b(1) = 0,
\end{equation*}
or \textbf{(II)} dependent on the chemoattractant, i.e.~$b \equiv b(m,c)$, with 
\begin{equation*}
    \partial_m b(m,c) <0, \quad \partial_c b(m,c) >0, \quad b(m,0) = 0,\quad b(1,c) = 0.
\end{equation*}
The division kernel has the form
\begin{equation}
\label{eq:ass:B}
B(y,m) = \frac 1 \epsilon \hat{B} \left( y, \frac{y-m}{\epsilon} \right), 
\end{equation}
where $\hat{B}, \partial_m \hat{B} \in C(\mathbb R^2)\cap L^1(\mathbb R^2)$, and 
\begin{equation} 
\label{eq:ass:B_integral}
\sup_{y\geq 0} \int_{-\infty}^\infty e^{|z|^2} \hat{B}(y,z) dz < \infty \quad \text{and} \quad 
\int_{-\infty}^\infty \hat{B}(y,z) dz = 1 \text{ for all } y \in [0,1].
\end{equation}
Moreover, we define
\begin{equation*}
    \bar{B}(y) := \int_{-\infty} ^\infty z \hat{B}(y,z) dz,  
\end{equation*}
i.e.~daughter cells from a mother of trait $y$ will be, on average, of trait $\bar{B}(y)$.
\item[(A4)] The functions $\omega_1(m), \omega_2(m) \in C^1([0,1])$ satisfy $0 \leq \omega_{i}(m) \leq 1$ for $i=1,2$ and
\begin{equation*}\begin{split}
    \omega_1(1) = 0, \quad \omega_2(0) = 0.
\end{split}
\end{equation*}
\item[(A5)] The cell  density
$p_\epsilon (x,t,m)$ is non-negative and bounded, for all $\epsilon >0$.
\end{itemize}
\end{assumption}

Similar to  \cite{lorenzi2022trade,lorz2011dirac}, we expect that $p_{\epsilon}$ can be approximated as
\begin{equation*}
p_\epsilon (x,t,m) \approx P(x,t)\delta(m - \bar{m}(x,t) ), \quad \text{for } \epsilon \to 0,
\end{equation*}
where $\delta(m-\bar m)$ is the Dirac Delta centered at $\bar m$.  We will call $\bar{m}(x,t)$ the dominant trait at $(x,t)$.

To characterise the properties of travelling wave solutions of \eqref{eq:system}, we will use the Hopf-Cole transformation $p_\epsilon(x,t,m) = \exp \left( u_\epsilon (x,t,m)/\epsilon \right)$, and denote by $p,u,c$ the leading order terms in the asymptotic expansions for $p_\epsilon, u_\epsilon, c_\epsilon$ in powers of $\epsilon$. First we shall derive the equation satisfied by the dominate trait $\bar{m}$. 

\begin{proposition} 
Under Assumptions~\ref{assum}, in the limit as $\epsilon \to 0$ the dominant trait $\bar{m}(x,t)$ is defined for $(x,t) \in \text{supp}(P)$ and evolves according to 
\begin{equation}
\label{eq:barm}
\partial_t \Bar{m} + A(\Bar{m}) \partial_x \Bar{m} \partial_x c =\left[  K_1 \omega_1(\bar{m}) c - K_2 \omega_2(\bar{m})\right] - \frac{\partial_m b(\bar{m},c)}{2\partial^2_{mm}u(\cdot,\cdot,\bar{m})}  - b(\bar{m},c) \bar{B}(\bar{m}),
\end{equation}
for $c$ satisfying \eqref{eq:epsto0}.
\end{proposition}
\begin{proof}
For the Hopf-Cole transformation  $p_\epsilon (x,t,m) = e^{u_\epsilon(x,t,m) / \epsilon}$, and for $x \in \text{supp}(P_\epsilon)$, we assume that both $u_\epsilon$ and the leading order in $\epsilon$ term  $u$ are strictly concave in $m$ and that $u$ has a unique maximum at $\bar{m}$, i.e.
\begin{equation*}\partial_m u (x,t,\bar{m}) = 0, \quad \quad \partial^2_{mm} u (x,t,\bar{m}) < 0.
\end{equation*}
Substituting the Hopf-Cole transformation in \eqref{eq:system} gives
\begin{equation}\label{eq:u_eps}
\begin{aligned}
       \partial_t u_\epsilon &=  D(m) \left[ (\partial_x u_\epsilon)^2 + \epsilon \partial^2_{xx} u_\epsilon \right] - A(m) \left[ \partial_x u_\epsilon \partial_x c_\epsilon + \epsilon \partial^2_{xx} c_\epsilon \right]  \\
       & + \partial_m u_\epsilon \left[ K_2 \omega_2(m) - K_1 \omega_1(m) c_\epsilon \right] + \epsilon [K_2\omega_2'(m) - K_1 \omega_1'(m) c_\epsilon] \\
        & -\gamma (P_\epsilon) - \frac{1}{2} b(m,c_\epsilon) \int_0^1 B(m,y) dy  + \int_0^1 b(y,c_\epsilon) B (y,m) e^{\frac 1 \epsilon[u_\epsilon (\cdot,\cdot,y)- u_\epsilon (\cdot,\cdot,m) ]} dy
        \end{aligned}
        \end{equation}
        and
        \begin{equation}
        \label{eq:c_eps}
            \begin{aligned}
         \partial_t c_\epsilon =  \epsilon d \partial^2_{xx} c_\epsilon - Q(P_\epsilon) c_\epsilon + \sigma K_2 \int_0^1 \omega_2(y) e^{\frac{u_\epsilon(\cdot,\cdot,y)}\epsilon} dy - \sigma K_1 c_\epsilon \int_0^1 \omega_1(y) e^{\frac{u_\epsilon(\cdot,\cdot,y)}\epsilon} dy.
\end{aligned}
\end{equation}
It is practical to make a change of variable $y \to z$, given by $y=m+\epsilon z$. Then, using assumptions \eqref{eq:ass:B} and \eqref{eq:ass:B_integral}, the first integral in \eqref{eq:u_eps} approaches $1$ in the limit as $\epsilon \to 0$. For the last integral in \eqref{eq:u_eps} we obtain
\begin{equation*}
\int_{-m/\epsilon}^{(1-m)/\epsilon} b(m + \epsilon z, c_\epsilon) \hat{B}(m+\epsilon z,z) \exp\Big(\frac{u_\epsilon(x,t,m+\epsilon z) - u_\epsilon (x,t,m) }\epsilon\Big) dz.
\end{equation*}
Since   $p_\epsilon$ is non-negative and bounded for all $\epsilon >0$,  we obtain from the form of the Hopf-Cole transformation in the limit the following property
\begin{equation}
\label{eq:constraintsu}
    u(x,t,\Bar{m}(x,t)) = \max_{m \in [0,1]} u (x,t,m) = 0, \quad  (x,t) \in \text{supp}(P).
    \end{equation}
Using these properties  and letting $\epsilon \to 0$ in \eqref{eq:u_eps} and \eqref{eq:c_eps}, we formally obtain 
\begin{equation} \label{eq:epsto0}
    \begin{split}
        \partial_t u = & D(m)  (\partial_x u)^2  - A(m) \partial_x u \partial_x c + \partial_m u \left[ K_2 \omega_2(m) - K_1 \omega_1(m) c \right] \\
        &  -\gamma(P) - \frac{1}{2} b(m,c) \underbrace{\int_{-\infty}^\infty \hat{B}(m,z) dz }_{=1} + b(m,c) \int_{-\infty}^\infty  \hat{B} (m,z) e^{z \partial_m u} dz,\\
        \partial_t c = &  - Q(P) c + \sigma \left[ K_2 \omega_1(\bar{m}) - K_1 c \omega_2 (\bar{m}) \right]P.
    \end{split}
\end{equation}
To take a limit in  the integrals in \eqref{eq:c_eps} we used 
\begin{equation*}
\int_0^1 \omega_i(y) e^{u_\epsilon (x,t,y)/\epsilon} dy \to \omega_i(\bar{m}(x,t)) P(x,t) \quad \text{ as } \epsilon \to 0, \; \text{ for a.e. }   t \geq 0, \; x \in \mathbb{R}, \; i=1,2,  
\end{equation*}
since  for $u_\epsilon (x,t,m) \leq 0$ and $u_\epsilon (x,t,\bar{m}) = 0$ the sequence  $e^{u_\epsilon (x,t,m)/\epsilon}$ converges to $P(x,t) \delta(m - \bar{m} (x,t) )$ as $\epsilon \to 0$. For simplicity, denote the reaction term in the equation for $u$ as
\begin{equation*}
R(m,u,c,P):=-\gamma(P) - \frac{1}{2} b(m,c)  + b(m,c) \int_{-\infty}^\infty  \hat{B} (m,z) e^{z \cdot \partial_m u(x,t,m)} dz.
\end{equation*}
The property \eqref{eq:constraintsu} implies
\begin{equation}
\label{eq:assumpDu}
\partial_x u (x,t,\Bar{m}) = 0, \quad \partial_t u (x,t,\Bar{m}) = 0, \quad \partial_m u (x,t,\Bar{m}) = 0 \quad \text{for } (x,t) \in \text{supp}(P).
\end{equation}
Evaluating \eqref{eq:epsto0} at $m=\bar{m}(x,t)$  and using \eqref{eq:constraintsu} and \eqref{eq:assumpDu} gives 
\begin{equation}
\label{eq:R=0}
-\gamma(P) + \frac{1}{2} b(\bar{m},c)  = 0 \quad \text{ for } \; (x,t) \in \text{supp}(P).
\end{equation}
Differentiating~\eqref{eq:epsto0} w.r.t.~$m$, evaluating the resulting equation at $m= \bar{m}$, and using \eqref{eq:assumpDu}, yield
\begin{equation}
\label{eq:seconddiffm}
\begin{aligned}
\partial^2_{tm} u (\cdot,\cdot,\bar{m}) = & -A(\bar{m}) \partial^2_{xm} u(\cdot,\cdot,\bar{m}) \partial_x c \\
& + \partial^2_{mm} u(\cdot,\cdot,\bar{m}) \left[ K_2\omega_2 (\bar{m}) - K_1 \omega_1 (\bar{m}) c\right] + \partial_m R(\bar{m},u,c,P),
\end{aligned}
\end{equation}
where
\begin{equation*}
\begin{aligned}
\partial_m R (\bar{m},c,u,P) = &\frac{1}{2} \partial_m b(\bar{m},c) 
 \\
 &+ b(\bar{m},c) \Big[  \int_{-\infty}^\infty \partial_m \hat{B} (\bar{m},z) dz + \partial_{mm}^2 u (\cdot,\cdot,\bar{m}) \int_{-\infty}^\infty z \hat{B}(\bar{m},z) dz \Big].
 \end{aligned}
\end{equation*}
Since  $\partial_m \hat{B}(m,z)$ is continuous in $m$ and $z$ then 
\begin{equation*}
\int_{- \infty}^{\infty} \partial_m \hat{B}(\bar{m},z) dz = \partial_m \int^{\infty}_{- \infty} \hat{B}(m,z)dz \Big| _{m= \bar{m}}=0.
\end{equation*}
Differentiating $\partial_m u(x,t,\bar{m}(x,t)) = 0$ with respect to  $t$ and $x$ yields
\begin{equation*} 
\partial^2_{tm} u (x,t,\Bar{m}) = - \partial^2_{mm} u(x,t,\Bar{m}) \partial_t \Bar{m}(x,t), \quad \partial^2_{xm} u (x,t,\Bar{m}) = - \partial^2_{mm} u(x,t,\Bar{m}) \partial_x \Bar{m}(x,t),
\end{equation*}
and finally, by substituting these into \eqref{eq:seconddiffm}, we obtain \eqref{eq:barm}.
\end{proof}

Now we characterise the properties of travelling wave solutions of~\eqref{eq:system} in the limit as~$\epsilon \to 0$. 
\begin{proposition}
\label{res:A1}
Let Assumptions~\ref{assum} be satisfied. If
\begin{equation}
\label{eq:ass:b}
\begin{split}
b(\bar{m},c) \bar{B}(\bar{m}) >& K_1 \omega_1(\bar{m}) c  -K_2 \omega_2(\bar{m}) \quad \text{ for $(x,t) \in $} \text{supp} (P),
\end{split}
\end{equation}
in the limit as  $\epsilon \to 0$, travelling wave solutions $p(z,m) = P(z) \delta(m-\bar{m}(z))$ and $c(z)$ with $\displaystyle{\lim_{z \to - \infty}} \bar{m}(z) = 0$, $\displaystyle{\lim_{z \to \infty}} c(z) = c_0$, where $z=x-vt$, for the re-scaled system~\eqref{eq:system} have a minimal wave speed 
\begin{equation}
\label{eq:wavespeed}
    v\geq v^* := \sup _{z \in \text{supp} (P) } \sqrt{A(\bar{m}) \left[ c(z) ( Q(P(z)) + \sigma P(z) K_1 \omega_1(\bar{m}) ) -  \sigma P(z) K_2 \omega_2(\bar{m}) \right] }.\end{equation}

These solutions have the properties
\begin{itemize}
    \item[(i)] $\bar{m}'(z) >0$,
    \item[(ii)] $c'(z) >0$, $\displaystyle \lim_{z \to -\infty } c(z) =0$.
\end{itemize} 
Considering the case \textbf{(I)} in (A3), $P(z)$ satisfies
\begin{itemize}
        \item[(iii)] $P'(z)<0$, $\beta P(z) = \frac{1}{2} b(\bar{m}(z))$ on $z \in \text{supp} (P) $ and 
    \begin{equation}
    \label{eq:Pequilib}
        \lim_{z \to - \infty} P(z) = \frac{1}{2 \beta} b(0),\end{equation}
\end{itemize}
whereas considering the case \textbf{(II)} in (A3), $P(z)$ satisfies
\begin{itemize}    
\item[(iii)]  $\text{supp} (P) \subset \text{supp} (c)$, $\beta P(z) = \frac{1}{2} b(\bar{m}(z),c(z))$, and $P(z)$ attains a maximum at $z^*$, where $z^*$ satisfies $\left[ \partial_m b(\bar{m},c) \bar{m}' + \partial_c b (\bar{m},c) c' \right]_{z = z^*} = 0$.
    \end{itemize} 

The leading edge of the travelling wave coincides with the transition point $\zeta \in \mathbb{R}$ between invaded ($P(z)>0$) and non-invaded region ($P(z)=0, c(z) = c_0$). From \eqref{eq:R=0}, we obtain that this is also the point where $\bar{m}(\zeta) = 1$.
\end{proposition}
\textit{Remark.} Using the assumption on $c_0$ and that $0 \leq c(z) \leq \max \left[ \frac{\sigma K_1 \omega_1 (\bar{m}) }{q + \sigma K_2 \omega_2 (\bar{m}) }, c_0 \right]$, an upper bound for the minimal travelling wave speed is
\begin{equation*}
    v^* \leq \sup_{z \in \text{supp}(P)} \sqrt{ \max_{m \in [0,1] } A(m) P(z) \sigma K_1 }.
\end{equation*}
\begin{proof}
Introducing the travelling wave ansatz with $z:= x - vt$ for  $v>0$ and 
\begin{equation*}
 u(x,t,m) = u(z,m), \quad c(x,t) = c(z), \quad 
 P(x,t) = P(z), \quad \bar m(x,t)=\bar{m}(z), 
 \end{equation*}
 from \eqref{eq:epsto0} we obtain for $u$ and $c$ the following equations
\begin{equation}
\label{eq:uandc}
    \begin{split}
        \partial_z u \left[ A(m) c' - v\right] &= D( m) ( \partial_z u)^2 + \partial_m u (K_2 \omega_2(m) -K_1 \omega_1(m) c) +R(m,u,c,P),\\
        -v c' &= -Q(P) c + \sigma [K_2 \omega_2(\bar{m}) - K_1 \omega_1(\bar{m}) c] P.    \end{split}
\end{equation}
Properties \eqref{eq:constraintsu} and \eqref{eq:assumpDu} on $u$ yield
\begin{align}
u(z,\bar{m}(z)) &= \max_{m \in [0,1]} u(z,m) = 0, \quad \text{for } z \in \text{supp}(P),\\
 \partial_z u(z,\Bar{m}(z)) &= \partial_m u (z,\Bar{m}(z)) = 0, \quad \text{for } z \in \text{supp}(P).
\end{align}
Equation \eqref{eq:R=0} is also valid in travelling wave coordinates,
\begin{equation}  \label{eq:R=0_TV}
    \gamma (P(z)) = \frac{1}{2} b(\bar{m}(z),c(z)), \quad \text{for} \quad z \in \text{supp} (P).
\end{equation}
Moreover, from equation \eqref{eq:barm} we obtain
\begin{equation} \label{eq:barm_travw}
    \Bar{m}' \left[ A(\Bar{m}) c' - v\right] = [K_1 \omega_1(\bar{m}) c-K_2 \omega_2(\bar{m}) ] - \frac{ \partial_m b(\bar{m},c) }{2 \partial^2_{mm} u (z,\bar{m})}  - b(\bar{m},c) \bar{B}(\bar{m}),
\end{equation}
for $z \in \text{supp}(P)$.  With the assumption on initial conditions~\eqref{eq:ass:c0}, i.e.\ $\displaystyle \lim_{z \to \infty} c(z)=c_0 > \frac{\sigma K_2}{q}$ and \eqref{eq:uandc}, we have $c'(z) \geq 0$, and at the back of the wave the chemoattractant will converge to its steady state, i.e.
$$
\lim_{z \to - \infty} c(z) = \lim _{z \to - \infty} \frac{\sigma K_2 \omega_2(\bar{m}(z)) P(z)}{Q(P(z)) + \sigma K_1 \omega_1 (\bar{m}(z)) P(z)}.
$$
We look for invading cell fronts with positive speed $v$ and monotonically increasing solutions ($\bar{m}'(z)>0$) to the differential equation \eqref{eq:barm_travw} subject to the asymptotic condition
\begin{equation}
\label{eq:mbar_limit0}
\lim_{z \to - \infty} \bar{m}(z) = 0,
\end{equation}
as we are interested in solutions with a more proliferative population in the back of the travelling wave. As highlighted in the introduction, traveling cell bands often exhibit a distinct spatial organization: more migratory cells dominate the leading edge, while proliferative cells are concentrated toward the rear. This specific arrangement is of particular interest due to its relevance in various biological applications.\\
This leads, together with $\omega _2(0) = 0$, see assumption~(A4), to
\begin{equation}\label{eq:neglimitc}
\lim_{z \to - \infty } c(z) = 0.
\end{equation}
Imposing
\begin{equation}
\label{eq:mbar_frac}
\bar{m}'(z) = \frac{[K_1 \omega_1(\bar{m}) c-K_2 \omega_2(\bar{m}) ] - \frac{ \partial_m b(\bar{m},c) }{2 \partial^2_{mm} u (z,\bar{m})} - b(\bar{m},c) \bar{B}(\bar{m})}{A(\bar{m})c'(z) - v} >0
\end{equation}
we obtain the minimal wave speed $v^*$
\begin{equation}
\label{eq:minwavespeed}
v^*:= \sup_{z \in \text{supp} (P)} \sqrt{A(\bar{m}) \left[ c(z) (Q(P(z))+ \sigma P(z) K_1 \omega_1(\bar{m}(z))) -  \sigma P(z) K_2 \omega_2(\bar{m}(z)) \right] },
\end{equation}
since condition~\eqref{eq:ass:b} implies that the numerator in~\eqref{eq:mbar_frac} is negative.

Furthermore, differentiating \eqref{eq:R=0} with respect to $z$ gives
\begin{equation}
\label{eq:P'}
P' = \frac{1}{2\beta} (\partial_m b(\bar{m},c) \bar{m}' + \partial_c b(\bar{m},c) c').
\end{equation}
In case \textbf{(I)} in assumption (A3), this leads to the monotonically decreasing behaviour of $P$
\begin{equation*}
P'(z) = \frac{1}{2 \beta} \partial_m b(\bar{m}(z)) \bar{m}'(z)<0.
\end{equation*}
In the case where cells rely on consumption of the chemoattractant for proliferation, i.e.~case \textbf{(II)} in assumption (A3), the total cell density attains a maximum at $z^*$, where $z^*$ is such that \eqref{eq:P'} evaluated at $z^*$ equals zero.
From \eqref{eq:R=0}, \eqref{eq:mbar_limit0} and \eqref{eq:neglimitc}, we can deduce that at the back of the wave the total cell density will converge to
\begin{equation*}
\lim_{z \to - \infty } P(z) = \frac{1}{2 \beta} b(0,0).
\end{equation*}
The dominant trait $\bar{m}$ increases from the back to the leading edge of the wave. From \eqref{eq:R=0} and condition \eqref{eq:ass:b}  we find that the unique point $\zeta \in \mathbb{R}$ where $\bar{m} (\zeta )=1$ coincides   with the point where
\begin{equation*}
P(\zeta ) = \frac{1}{2 \beta} b(1,c_0) = 0  \quad \text{ or } \quad P(\zeta ) = \frac{1}{2 \beta} b(1) = 0,
\end{equation*}
in case  \textbf{(II)} or \textbf{(I)} respectively, 
which is the separation point between invaded and non-invaded region, i.e.~$P(z)=0$, $c(z) = c_0$ for all $z > \zeta$.
\end{proof} 

\section{Numerical results}
\label{sec:numres}

We consider system \eqref{eq:system} on a space-time-trait domain $(0,L) \times (0,T) \times (0,1)$, where $L>1$ and $T$ were chosen large enough to allow the evolution of the solution close to a travelling wave. We consider homogeneous Neumann boundary conditions in the spatial variable for both  $p_\epsilon$ and $c_\epsilon$. For readability we will drop the $\epsilon$ subscript in this section. For details of the numerical method we refer the reader to the Supplementary Information. 

At the start of the simulation, cells are taken to be concentrated on the left boundary of the domain, i.e.\ close to $x=0$, and to be uniformly distributed in the membrane occupancy $m$. Note that this means that, in addition to the initial free ligand concentration $c_0$, there is also initial membrane-bound ligand in the system. Initially, the free ligand concentration is set to a constant value over the whole domain. Summarising, we choose
\begin{equation}
\label{Ch2:eq:IC}
p(x,0,m) \equiv p(x,0) := \begin{cases} p_0 \quad \text{for} \quad x\leq 1
\\ 0 \quad \text{elsewhere}, \end{cases} \quad c(x,0) \equiv c_0 \text{ for } x \in [0,L]. 
\end{equation}
These initial conditions represent a scenario in which cells are exposed to a chemoattractant rich environment, e.g. describing \textit{in vitro} experiments in which cells or bacteria are placed on one side of a rectangular petri dish and an attractant is uniformly distributed. The parameter values used for the simulations are stated in the figure captions. 

We consider: 
\begin{itemize}
    \item Three different setting modes \textbf{(I/II/III)} for different couplings between ligand and cell dynamics, where (\textbf{I-II}) correspond to the assumptions in (A3);
    \item Two different setting modes \textbf{(A/B)} for distinguishing between different combinations of the trait-dependent coefficient functions;
    \item Three different setting modes \textbf{(1/2/3)} for the type of division kernel (cf. Figure~\ref{fig:schematic_division}).
\end{itemize} 
The simulation results are denoted by the three indicators, e.g. $\textbf{IA1}$, respectively. Note that not all of these scenarios fulfill the assumptions for the analytical result to hold.\\
\noindent
\textbf{Ligand-cell dynamic couplings (I/II/III):} 

\begin{itemize}
    \item[\textbf{(I)}] The ligand acts only as a  chemoattractant and is degraded by the cells.  Here, the rate of division does not depend on ligand availability, i.e. cells do not rely on the presence of ligands to proliferate. We choose a standard form for the division rate
    \begin{equation}
    \label{Ch2:res:b(m)}
        b(m) := b_0 (1-m).
    \end{equation}

    \item[\textbf{(II)}] The ligand acts as both a nutrient and a chemoattractant. Hence, ligands are now considered to be essential for the cells to proliferate. Specifically, we choose
    \begin{equation}
    \label{Ch2:res:b(m,c)}
        b(m,c) := b_0 (1-m) c.
    \end{equation} 
    \item[\textbf{(III)}] The ligand acts only as a chemoattractant, as in case \textbf{(I)}, but is not degraded by the cells (i.e. $Q(P) \equiv 0$). This case holds particular relevance, as the sole processes that modulate free ligand in the environment are attachment and detachment.
    \end{itemize}

\noindent
\textbf{Trait-dependent coefficient function settings (A/B):}
    \begin{itemize}
    \item[\textbf{(A)}] Chemotactic sensitivity and division rate are independent of the cell trait, i.e. $A(m) \equiv A$ and $b(m,c) \equiv b(c)$. In particular, we choose $A=1$ and $b(c) \equiv b = 0.75$ for \textbf{(I)} or $b(c)=0.75 c$ for \textbf{(II)}. In this case all cells have equal capacity to divide and migrate.
    \item[\textbf{(B)}] Trade-off between the chemotactic sensitivity and the rate of division. We choose $A(m) := 2m$ and $b(m)$ as stated in \eqref{Ch2:res:b(m)} and \eqref{Ch2:res:b(m,c)}. Here, cells with lower membrane occupancy are more prone to divide and less migratory, and vice versa for cells with a higher membrane occupancy. Biologically, this could reflect cells with a limited energy, which must balance dividing and migrating activities, see e.g. \cite{wang2019diverse, waugh2008interleukin}.
    \end{itemize}

    \noindent
    \textbf{Division settings (1/2/3):}
    \begin{itemize}
    \item[\textbf{(1)}] We consider symmetric division of type $m \to m/2 + m/2$ (bound ligand split equally between two daughter cells). In particular, we choose the kernel 
    \begin{equation}
    \label{Ch2:eq:divk_B1}
        B_1(m,y) := 
        \frac{V_1(m,y)}{\|V_1(m,\cdot)\|}, \; \; \text{ with } \;
      V_1(m,y) := \frac{1}{\sqrt{0.02 \pi }} \exp\Big(\!\! -\frac{(m/2-y)^2}{0.02}\Big),  
    \end{equation}
    where $\|v(m, \cdot)\| = \int_0^1 v(m,y) dy$   and the first variable, $m$, is the mother trait and the second variable, $y$, denotes the daughter trait.
    \item[\textbf{(2)}] We consider an asymmetric division of type $m \to m + 0$ (one daughter obtains all the ligand) through the kernel 
    \begin{equation}
    \label{Ch2:eq:divk_B2}
    B_2(m,y) := \frac{2}{3} 
     \frac{V^1_2(m,y)}{\| V_2^1(m,\cdot)\|}  + \frac{1}{3} 
      \frac{V^2_2(m,y)}{\| V_2^2(m,\cdot)\|} ,
    \end{equation} where
    \begin{equation*}
        V^1_2 (m,y) := \frac{1}{\sqrt{0.02 \pi }} \exp\Big(\!\! -\frac{(m-y)^2}{0.02}\Big) \; \text{ and } \;
        V^2_2 (m,y) := \frac{1}{\sqrt{0.02 \pi }} \exp\Big(\!\! -\frac{y^2}{0.02}\Big).
    \end{equation*}
    \item[\textbf{(3)}] Finally, we consider consumptive division of type $m \to 0 + 0 $ (all ligand consumed during division) through the kernel 
    \begin{equation} \label{Ch2:eq:divk_B3} 
    B_3(m,y): = 
     \frac{ V_3(m,y)}{\|V_3(m,\cdot)\|}, \;\; \text{ where } \;  \;
      V_3(m,y):= \frac{1}{\sqrt{0.02 \pi }} \exp \Big(\!\! -\frac{y^2}{0.02}\Big).  
    \end{equation}
    \end{itemize}
For schematics of the division processes see Figure~\ref{fig:schematic_division}. 

\textit{Remark.} We note that in the limit $\epsilon \to 0$, where the analysis applies, cell division and death become very fast and dominate detachment and attachment processes. Therefore, cells at the front of travelling waves, where they are exposed to higher ligand concentrations, will continue to attach ligands. However, cell division is fast and a mother cell with fully occupied membrane will divide into two daughter cells of lower membrane occupancy. Simulations have confirmed that when $\epsilon <<1$ is small enough, the mean trait will settle at a lower value throughout the travelling wave. We observed more diverse dynamical behaviours in the trait-space for parameter regimes with larger $\epsilon$ (where $\epsilon < 1$), and numerical simulations explore cases that extend beyond the scope of the analytical results. While analytical results extend to trait-dependent diffusion coefficients $D(m)$, in numerical simulations we assume $D(m)$ to be constant.

We will explore some biologically interesting questions:
\begin{itemize}
    \item In the case of equal abilities to proliferate and divide (case \textbf{A}), how does the cell population structure when exposed to the ligand concentration?
    \item When there is a trade-off between proliferation and migration (case \textbf{B}), does the cell population structure differently?
\end{itemize}

\subsection{Invading front in trade-off \textbf{(IB1)} vs no trade-off \textbf{(IA1)} scenarios}
\label{Ch2:appx:I}
 
In this case, cells do not rely on the presence of ligands to proliferate. Hence, we expect the cell population to equilibrate at the back of the travelling wave, despite ligand density vanishing in that region. We compare the two settings \textbf{A} and \textbf{B}. In the former, the chemoattractant sensitivity and division rate are homogeneous across the cell population. Hence, trait-structuring can only arise due to attachment/detachment of ligands and division. In case \textbf{B}, cells with higher membrane occupancy have the advantage of following the ligand gradient more efficiently; whereas cells with lower occupancy divide faster. Hence, we expect some trait-dependent sorting of the cells due to the trade-off. 

In both cases, see Figures~\ref{fig:IA1}-\ref{fig:IB1}, we indeed observe an invading front with positive travelling wave speed. As expected, the cells at the back of the wave continue to divide despite the lack of ligand. The total cell density decreases over the observed domain, as predicted. The ligand density increases from invaded to non-invaded region, due to its degradation by cells. In Figures~\ref{fig:IA1density} and ~\ref{fig:IB1density}, contour lines of the scaled cell population density $p(t,x,m)/P(t,x)$ are plotted in the $(x,m)$-domain to illustrate the distribution over the traits of the population. The dominant trait is increasing from the back to the front of the wave, in the case with the trade-off more than without. A reason for this might be the spatial sorting of different cell types in case \textbf{(IB1)}: more migratory cells accumulate at the front of the wave, while less migratory, but faster dividing cells fall behind and dominate the back of the wave. Hence, there is a segregated structuring of traits in the trade-off scenario.
The measured wave speed is compared to the analytically obtained minimal wave speed $v^*$ \eqref{eq:wavespeed} in Table~\ref{tab:wavespeed}. Note that the measured wave speed is higher than the evaluated minimal wave speed. This could be due to the evaluated speed applying in the limit $\epsilon \to 0$, whereas numerical simulations require an $\epsilon$ which is still relatively large (i.e. $\epsilon = 0.1$ for the simulation results in Figure~\ref{fig:IIAB1}).

\subsection{Travelling pulse with trade-off \textbf{(IIB1)} vs no trade off \textbf{(IIA1)}}
\label{Ch2:appx:II}
In contrast to the previous scenario, here cells rely on the presence of free ligand to proliferate. As the ligand concentration decreases, we expect the total cell density to decrease accordingly. As predicted by our analytical findings, the travelling wave solutions for the ligand are monotonically increasing. 
A significant difference to the cases presented before lies in the shape of the total cell density, which takes the shape of a travelling pulse. Specifically, it decreases at the back of the wave due to the lack of nutrient, has a maximum close to the front, before sharply decreasing at the border with the non-invaded region. As before, the dominant trait increases from the back to the front of the wave, see Figures~\ref{fig:IIA1density}-\ref{fig:IIB1density}. 
In the case \textbf{IIB1} (Figure~\ref{fig:IIB1}), where a trade-off occurs, the maximum of $P$ is higher than in the case \textbf{IIA1} with coefficients independent of $m$ (Figure~\ref{fig:IIA1}).

\subsection{Invading front without ligand decay and different division kernels \textbf{(IIIB 1-3)}}
\label{Ch2:appx:III}
A particularly relevant scenario not covered by the analytical results arises when the ligand does not decay. Unlike in previous cases, the ligand gradient in this setting is generated solely through the processes of ligand attachment and detachment. This fundamental change introduces new dynamics: when $Q(P_\epsilon) = 0$, and in the limit $\epsilon \to 0$, the monotonicity of the travelling wave solution for the ligand component $c$ is no longer guaranteed. 

In addition to this novel ligand setting, we also introduce two alternative cell division strategies that deviate from the assumptions in Section~\ref{sec:TW}.

Our numerical results reveal travelling wave solutions in scenarios \textbf{(IIIB1)}, \textbf{(IIIB2)} and \textbf{(IIIB3)}, as shown in Figures~\ref{fig:IIIB1}-\ref{fig:IIIB2} and \ref{fig:IIIB3_SI}. In these cases, the total cell density stabilizes at a constant value at the back of the wave. This outcome aligns with expectations, as the growth dynamics are independent of nutrient availability, resulting in equilibrium where growth and decay terms balance. Toward the front of the wave, the total cell density decreases, while the dominant trait appears to increase from the back to the front of the wave. Altogether, the behaviour of the system in these scenarios bears a resemblance to the dynamics observed in Case \textbf{(I)}. This suggests that the qualitative properties predicted by the analytical framework extend to cases outside its strict assumptions, thereby demonstrating broader applicability.\\ The choice of division strategy significantly influences the shape of the travelling wave. For instance, using division kernel 2 defined via \eqref{Ch2:eq:divk_B2} (from $m$ to $m + 0$) results in a greater diversity of cell membrane occupancies at the front. This, in turn, creates a more heterogeneous composition amongst the front cells, each exhibiting distinct migration speeds due to the underlying assumptions of the model.

\begin{table}[h]
    \centering
    \begin{tabular}{c|c|c}
           Cases & Minimal wave speed $v^*$& Measured wave speed \\ \hline \hline
        \textbf{IB1} & $1.13$ & $1.44 \pm 0.02$ \\
        \textbf{IIB1} & $0.73$ & $3.13 \pm 0.09$\\
        
\end{tabular}
    \caption{Minimal wave speed $v^*$ \eqref{eq:wavespeed} for cases that satisfy the assumptions made for the formal analysis; Mean and standard deviation of measured wave speed}
    \label{tab:wavespeed}
\end{table}

\begin{figure}
\centering
\begin{subfigure}[b]{0.32\textwidth}
\centering
\includegraphics[width=\textwidth]{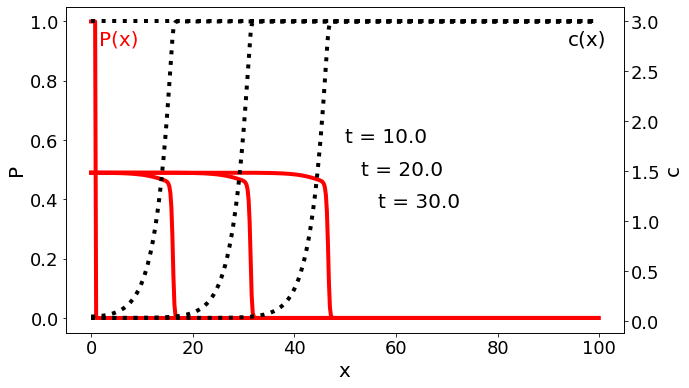}
\caption{\textbf{IA1}} \label{fig:IA1}
\end{subfigure}
\hfill
\begin{subfigure}[b]{0.32\textwidth}
\centering
\includegraphics[width=\textwidth]{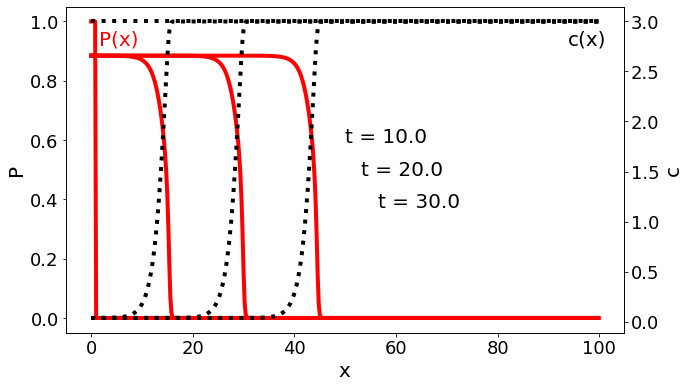}
\caption{\textbf{IB1}}
\label{fig:IB1}
\end{subfigure}
\hfill
\begin{subfigure}[b]{0.32\textwidth}
\centering
\includegraphics[width=\textwidth]{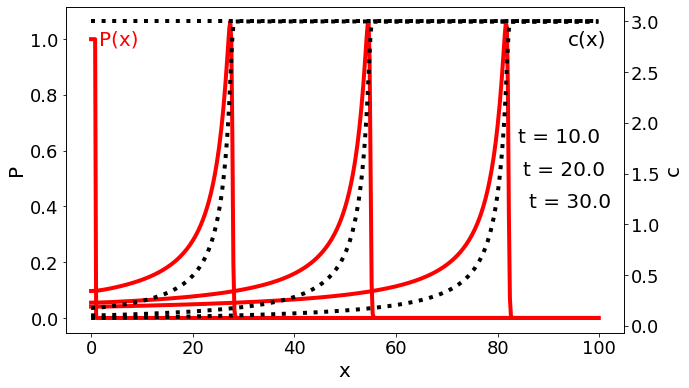}
\caption{\textbf{IIA1}}
\label{fig:IIA1}
\end{subfigure}
\hfill
\begin{subfigure}[b]{0.32\textwidth}
         \centering
         \includegraphics[width=\textwidth]{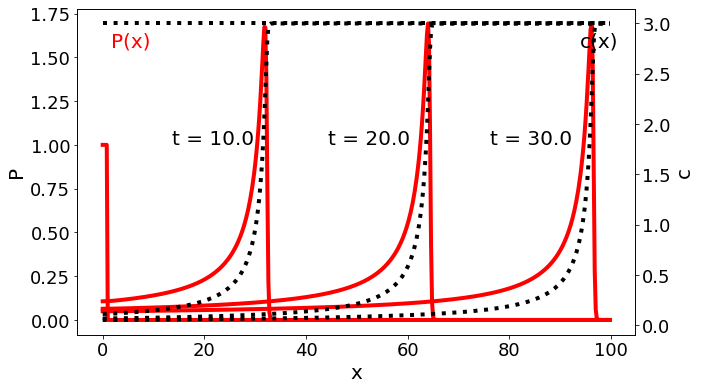}
         \caption{\textbf{IIB1}}
         \label{fig:IIB1}
         
     \end{subfigure}
\hfill
\begin{subfigure}[b]{0.32\textwidth}
         \centering
         \includegraphics[width=\textwidth]{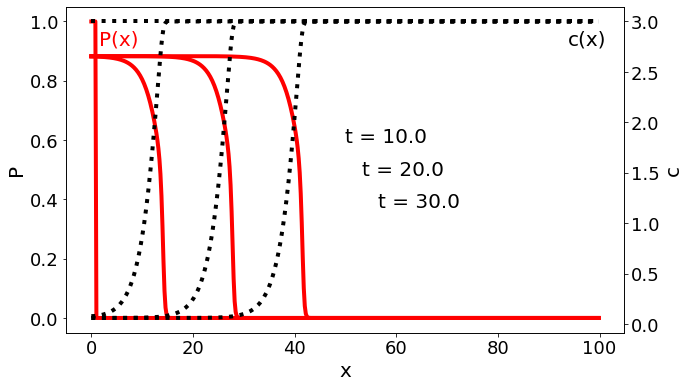}
         \caption{\textbf{IIIB1}}
         \label{fig:IIIB1}
     \end{subfigure}
\hfill
\begin{subfigure}[b]{0.32\textwidth}
         \centering
         \includegraphics[width=\textwidth]{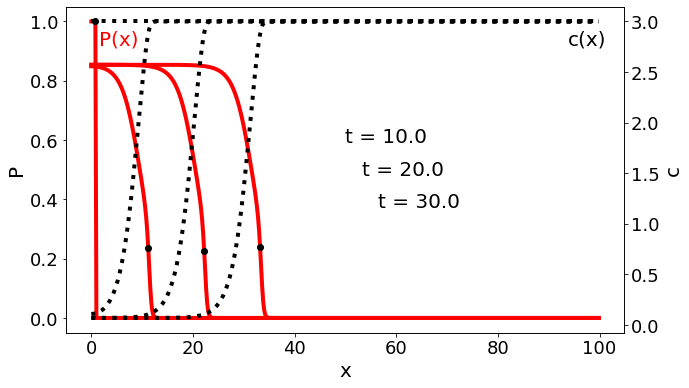}
         \caption{\textbf{IIIB2}}
         \label{fig:IIIB2}
     \end{subfigure}
\hfill 
\begin{subfigure}[b]{0.32\textwidth}
         \centering
         \includegraphics[width=\textwidth]{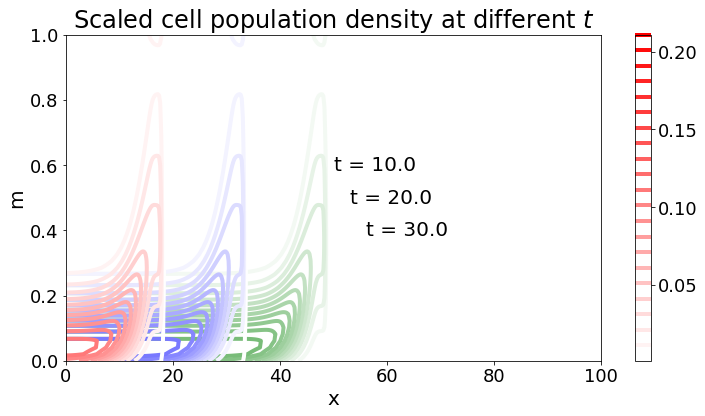}
         \caption{\textbf{IA1}}
         \label{fig:IA1density}
        
     \end{subfigure}
\hfill
\begin{subfigure}[b]{0.32\textwidth}
         \centering
         \includegraphics[width=\textwidth]{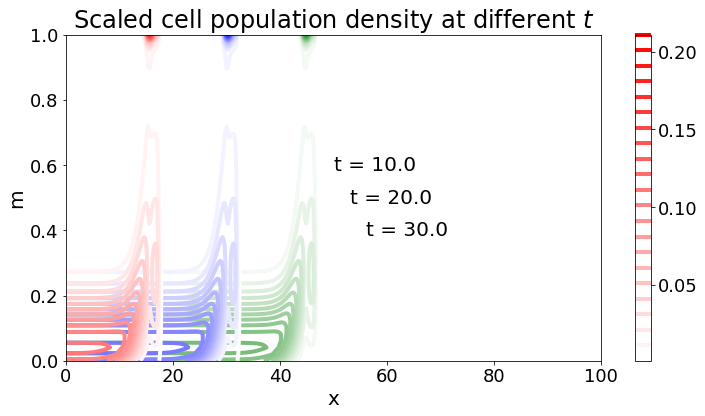}
         \caption{\textbf{IB1}}
         \label{fig:IB1density}
        
     \end{subfigure}
\hfill
\begin{subfigure}[b]{0.32\textwidth}
         \centering
         \includegraphics[width=\textwidth]{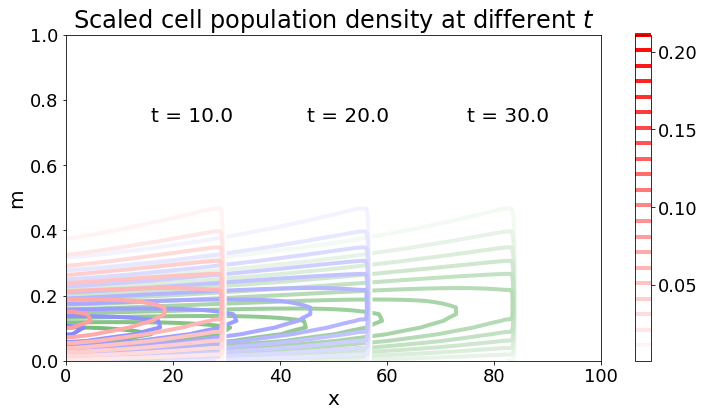}
         \caption{\textbf{IIA1}}
         \label{fig:IIA1density}
        
     \end{subfigure}
\hfill
\begin{subfigure}[b]{0.32\textwidth}
         \centering
         \includegraphics[width=\textwidth]{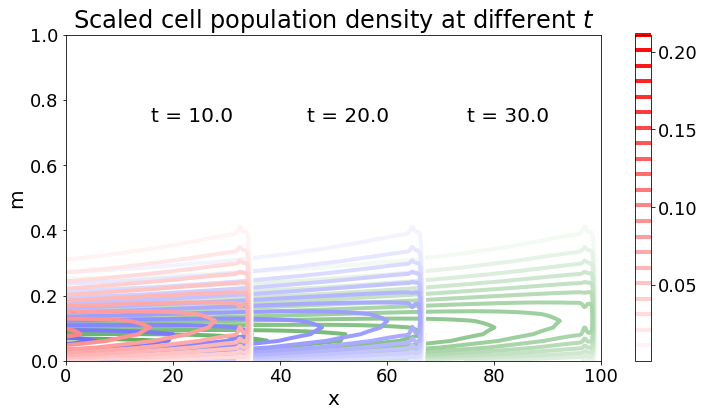}
         \caption{\textbf{IIB1}}
         \label{fig:IIB1density}
        
     \end{subfigure}
\hfill
\begin{subfigure}[b]{0.32\textwidth}
         \centering
         \includegraphics[width=\textwidth]{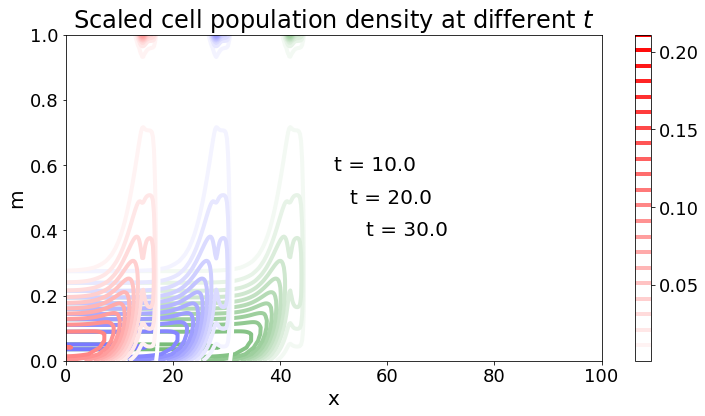}
         \caption{\textbf{IIIB1}}
         \label{fig:IIIB1density}
        
     \end{subfigure}
\hfill
\begin{subfigure}[b]{0.32\textwidth}
         \centering
         \includegraphics[width=\textwidth]{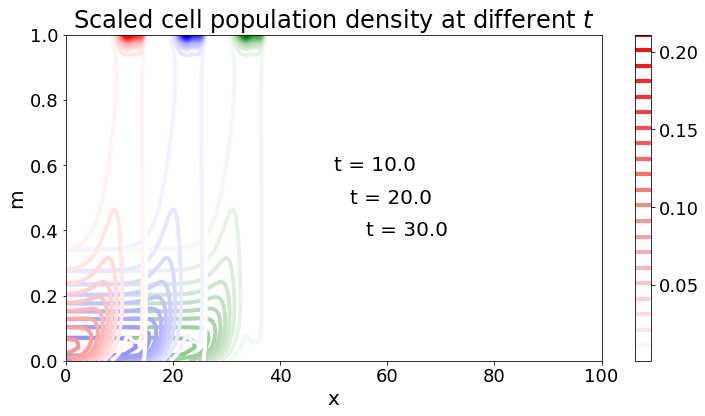}
         \caption{\textbf{IIIB2}}
         \label{fig:IIIB2density}
\end{subfigure}

\caption{(a-f): Total cell density (red lines) and concentration of chemoattractant (black dotted line) in space, at different time points for different scenarios; (g-l): Scaled cell population density $p(x,t,m)/P(x,t)$ at different time points in the $(x,m)$ domain. Parameters and functions used: $D(m) \equiv 1$, $b_0=1.5$, $\beta=0.75$, $K_1 = 0.75$, $K_2 = 0.1$, $d=1$, $q=0.5$, $\sigma =1$, $\omega_1(m) := \frac{1-m}{0.1+1-m}$, $\omega_2(m) := \frac{m}{0.1+m}$, $\epsilon = 0.1$, $B(m,y)=B_1(m,y)$ with $B_1$ defined via \eqref{Ch2:eq:divk_B1}; (a,c,g,i) $A(m)\equiv 2$, $b(m) := b_0c$; (b,d,e,f,h,j,k,l) $A(m):=2m$, $b(m) :=b_0(1-m)c$. }
        \label{fig:IIAB1}
\end{figure}

\section{Discussion} \label{sec:discussion}
We have formally derived a Keller-Segel type model that integrates a continuous trait structure in a heterogeneous cell population with the spatio-temporal evolution of a chemoattractant. This model was constructed at the individual scale and subsequently coarse-grained to derive a system of non-local partial differential equations. The continuous trait variable captures the proportion of receptors occupied by ligands on the cell membrane, which directly influences growth, diffusion, and chemotactic sensitivity. Changes in this trait arise from ligand attachment or detachment from membrane receptors, while a growth-fragmentation term models the inheritance of receptor occupancy during cell division. Additionally, the ligand dynamics are described by a diffusion-reaction equation that accounts for ligand removal and release during the attachment-detachment processes. \\
Using a Hopf-Cole transformation, we derived properties of travelling wave solutions under the assumption of time-scale separation, with slow diffusion and fast cell division/ death processes. These analytical results provide insights into the minimal wave speed and qualitative behaviour of the solutions, highlighting the influence of trait dynamics on popupation migration. The formal analysis underscored the relevance of trait-mediated processes in shaping wave propagation.\\
Our numerical simulations corroborate the analytical findings in scenarios that align with the assumptions of the formal calculations. Specifically, the numerical results qualitatively match the predicted properties of travelling wave solutions, with the measured wave speed consistently exceeding the calculated minimal wave speed. This agreement validates the theoretical framework and emphasises its applicability in capturing key dynamics of the system.\\
Beyond the cases covered by the assumptions underlying the formal analysis, the numerical experiments explore parameter regimes that deviate from the time-scale separation hypothesis. These simulations reveal richer dynamical behaviour in the trait-space and underscore the potential for more complex interactinos between diffusion, growth-fragmentation, and chemotactic sensitivity. The numerical results thus extend the scope of the study and provide a foundation for future investigations into parameter-dependent dynamics.\\
The model presented is framed in an abstract context to highlight general properties of migratory cell populations that modulate their behaviour based on ligand availability and receptor occupancy. This framework has broad applicability, including neural crest migration, where cells are often classified as ``leader'' or ``follower'' types. Previous agent-based models, such as those in \cite{schumacher2019neural}, have explored continuous cell states. By incorporating a continuum of cell states, our model offers a complementary perspective that can enhance the understanding of such systems. For instance, it provides a continuous representation of cell density along the trait-dimension, effectively bridging the gap between discrete-in-trait and continuum modelling approaches and offering a more nuanced understanding of cell population dynamics.\\
Our approach could also extend to cancer metastasis models, particularly those exploring the ``go-or-grow'' hypothesis, which differentiates between migratory and proliferative cell types \cite{giese1996dichotomy}. By introducing a continuous trait structure, this framework could generalise the dichotomous models into a continuous description, facilitating a more nuanced analysis of the trade-offs between migration and proliferation. \\
While existing models have considered trait mutations and selection, our model introduces an advection term in the trait space, with a coefficient modulated by the availability of an external component. This innovation builds on techniques from previous studies \cite{bouin2012invasion, lorenzi2022trade, macfarlane2022individual}, enabling the derivation of travelling wave properties under time-scale separation assumptions. Furthermore, our model aligns with the broader class of structured diffusion-advection-reaction equations, which have been extensively studied in age- and size-structured populations \cite{auger2008structured}.\\
A key assumption in our analysis is the separation of time scales, which simplifies the derivation of travelling wave solutions. Future work could relax this assumption to investiage whether travelling wave solutions persist under alternative parameter settings. Moreover, extending the model to incorporate additional biological complexities, such as stochastic effects or spatial heterogeneity, could provide deeper insights into the dynamics of trait-mediated population migration. These extensions would further bridge the gap between theoretical predictions and experimental observations, enhancing the utility of the model in biological applications.

\section*{Acknowledgments}
VF was supported by the EPSRC Centre for Doctoral Training in Mathematical Modelling, Analysis and Computation (MAC-MIGS) funded by the UK Engineering and Physical Sciences Research Council (grant EP/S023291/1), Heriot-Watt University and the University of Edinburgh. KJP acknowledges ``Miur-Dipartimento di Eccellenza’’ funding to the Dipartimento di Scienze, Progetto e Politiche del Territorio (DIST). TL gratefully acknowledges support from the Italian Ministry of University and Research (MUR) through the grant PRIN 2020 project (No. 2020JLWP23) ``Integrated Mathematical Approaches to Socio-Epidemiological Dynamics’’ (CUP: E15F21005420006) and the grant PRIN2022-PNRR project (No. P2022Z7ZAJ) ``A Unitary Mathematical Framework for Modelling Muscular Dystrophies’’ (CUP: E53D23018070001) funded by the European Union – NextGenerationEU. KJP and TL are members of INdAM-GNFM. The authors would like to thank the Isaac Newton Institute for Mathematical Sciences, Cambridge, for support and hospitality during the programme Mathematics of Movement where work on this paper was undertaken.



\appendix

\section{Numerical implementation}
\label{app:num}
For the numerical implementation we used a standard finite volume scheme to discretize the space and membrane occupancy dimensions, see for example \cite{hundsdorfer2003numerical}. The resulting system of ODEs for the time evolution was then solved by using the pre-implemented explicit Runge-Kutta method of order 5(4) from the python package \textit{scipy.integrate.solve\_ivp}.

The space and membrane occupancy intervals ($[0,L]$ and $[0,1]$, respectively) were discretised in the following way: let $N+1$ and $M+1$ denote the number of discretisation points (including both endpoints) of the space and occupancy interval, respectively. Then the full intervals are divided into $N $ and $M$ intervals of length $\Delta x = \frac{L}{N}$ and $\Delta m = \frac{1}{M}$, respectively.  We use $L=100$, $T=30$, for the simulation window $(x,m,t)\in [0,L]\times[0,1]\times[0,T]$, $N+1 = 500$ discretisation points in space and $M+1 = 50$ in the dimension for membrane occupancy, resulting in $\Delta x = 0.2$ and $\Delta m = 0.02$. 
Let $p_{i,j}(t) := p(x_i, t, m_j) = p(i \Delta x, t, j \Delta m )$ denote the cell population density at position $x_i$ with membrane occupancy $m_j$, and $c_i (t):= c(x_i, t) = c(i \Delta x, t)$ the chemoattractive concentration at position $x_i$. The total cell density is defined as $P_i(t) := \Delta m \sum_{j=0}^{M-1} p_{i,j}(t)$.

For $0 \leq i < N$, let $f^p_{i+1/2, j}(t)$ denote the flux of the population with membrane state $m_j $ from the $x_i$ to $x_{i+1}$ at time $t$,
\begin{equation}
\label{AppxTW:eq:FluxPx}
    f_{i+1/2,j}^p (t) = \begin{cases}  - \epsilon  \frac{1}{\Delta x } D^j (p_{i+1,j} (t) - p_{i,j} (t) ) + \frac{1}{\Delta x } A^j (c_{i+1} (t) - c_i(t) ) p_{i,j} (t)\\
    \quad \text{if} \quad c_{i+1}(t) - c_i(t) >0, \\
   - \epsilon \frac{1}{\Delta x } D^j (p_{i+1,j}(t) - p_{i,j}(t)) + \frac{1}{\Delta x } A^j (c_{i+1}(t) - c_i (t) ) p_{i+1,j} (t) \\
   \quad \text{if} \quad c_{i+1}(t) - c_i (t) < 0,
\end{cases}
\end{equation}
for $j=1,\dots,M$, where $D^j = D(m_j)$ and $A^j = A(m_j)$. To enforce homogeneous Neumann boundary conditions we define $f^p_{-1/2,j} (t) = 0 $ and $f^p_{N+1/2,j} (t)= 0 $ for $0 \leq t \leq T$, $j=1,\dots,M$.
Similarly, for $0 \leq j < M$ the flux of the population located at $x_i = i\Delta x$ from membrane state $m_j$ into $m_{j+1}$ is
\begin{equation}
\label{AppxTW:eq:FluxPm}
    f^p_{i,j+1/2} (t) = K_1 \omega_1(m_j) c_i p_{i,j} - K_2 \omega_2(m_{j+1}) p_{i,j+1},
\end{equation}
for $i=1,\dots,N$, where $\omega_1, \omega_2 \in C(0,1)$ with $\omega_1 (1) = 0$ and $\omega_2 (0) = 0$. This, together with $p_{i,j}(0) = 0 $ for $j<0$ and $j>M$ ensures $f_{i,-1/2}(t) = 0 $ and $f_{i,M+1/2} (t) = 0 $, for $i=1,\dots,N$, so that the scheme is mass conserving. The flux for the chemoattractive substance from $x_i$ to $x_{i+1}$ at time $t$ is
\begin{equation}
\label{AppxTW:eq:FluxC}
f^c_{i+1/2} = - \epsilon d \frac{1}{\Delta x } (c_{i+1}(t)- c_i(t) ), 
\end{equation}
for $i=1,\dots,N$, with boundary conditions $f^c_{-1/2} (t) = 0 $ and $f^c_{N+1/2} (t) = 0$ for $0 \leq t \leq T$.
The kinetics with the non-local terms are discretised as follows:
\begin{equation}
\begin{split}
    k^p_{i,j} (t) &= \frac{1}{\epsilon } \Big( - \gamma (P_i(t)) p_{i,j}(t) + \Delta m \sum_{z=0}^{M} b(m_z,c_i(t)) B(m_z,m_j) p_{i,z} (t) \\
    &\quad \qquad  - \frac{1}{2} p_{i,j} (t) b(m_j,c_i(t)) \Delta m \sum_{z=0}^{M} B(m_j,m_z) \Big), \\
    k^c_i (t) & = \sigma K_2 \Delta m \sum_{z=1}^{M} p_{i,z} (t)  - \sigma K_1 c_i (t) \Delta m \sum_{z=0}^{M-1} p_{i,z} (t) - Q(P_i(t)) c_i(t).
\end{split}
\end{equation}

The corresponding ODE system then reads
\begin{equation}
\begin{split}
\frac{d}{dt} p_{i,j}(t) &= \frac{1}{\Delta x } (f^p_{i-1/2,j}(t) - f^p_{i+1/2,j}(t)) + \frac{1}{\Delta m } (f^p_{i,j-1/2} (t) - f^p_{i,j+1/2} (t))  + k^p_{i,j} (t), \\
\frac{d}{dt} c_i(t) &= \frac{1}{\Delta x } (f^c_{i-1/2}(t) - f^c_{i+1/2}(t)) + k^c_i(t),
\end{split}
\end{equation}
for $0 \leq i \leq N$, $0 \leq j \leq M$ and with appropriate initial conditions $p_{i,j}(0) \geq 0 $, $c_i(0) \geq 0$.

\section{Supplementary figures}

\begin{figure}[h!]

\centering
    \begin{subfigure}[b]{0.32\textwidth}
         \centering
         \includegraphics[width=\textwidth]{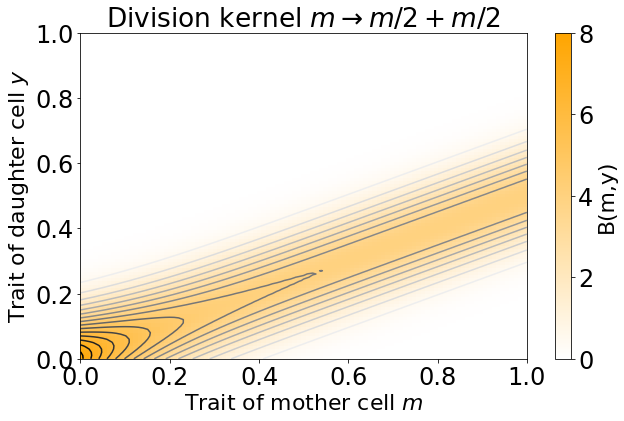}
         \caption{$B_1$}
         \label{fig:B1}
     \end{subfigure}
     \hfill
     \begin{subfigure}[b]{0.32\textwidth}
         \centering
\includegraphics[width=\textwidth]{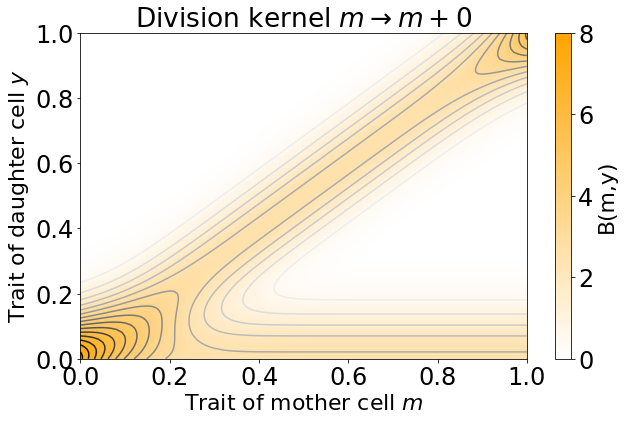}
         \caption{$B_2$}
         \label{fig:B2}
     \end{subfigure}
     \hfill
     \begin{subfigure}[b]{0.32\textwidth}
         \centering
         \includegraphics[width=\textwidth]{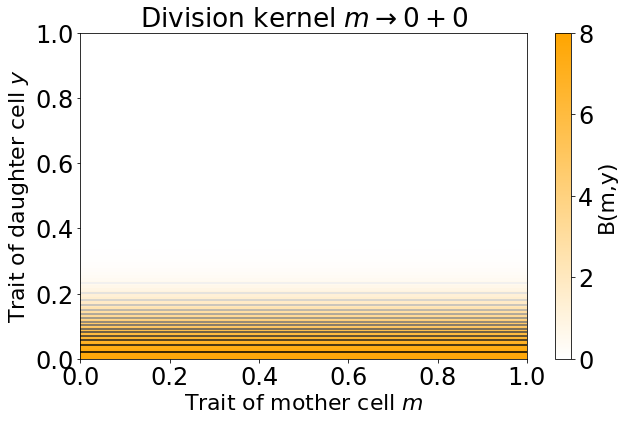}
         \caption{$B_3$}
         \label{fig:B3}
     \end{subfigure}
\caption{  Division kernels for: (a) Symmetric division (cf.\ definition (5.4) for the division kernel $B_1$) - the bound ligands are equally distributed between two daughter cells; (b) Asymmetric division (cf.\ definition (5.5) for the division kernel $B_2$) - one daughter cells inherits all the bound ligands, the other daughter cells has an empty membrane; (c) Consumptive division (cf.\ definition (5.6) for the division kernel $B_3$) - each cell divides into two cells with trait $0$ and bound ligands are consumed during division}
\label{fig:schematic_division2}
\end{figure}

\begin{figure}[h!]
\centering
\begin{subfigure}[b]{0.45\textwidth}
\centering
\includegraphics[width=\textwidth]{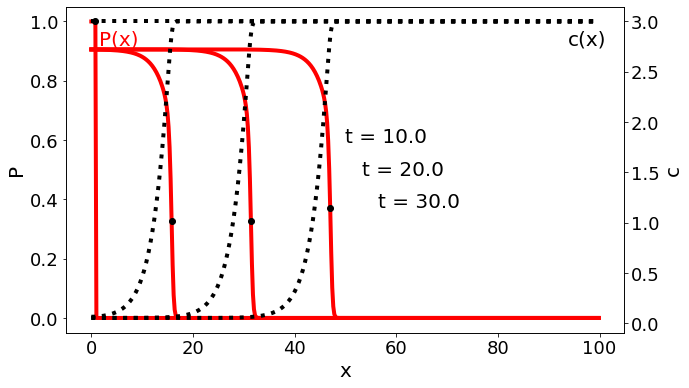}
\caption{\textbf{IIIB3}} \label{fig:IIIB3}
\end{subfigure}
\hfill
\begin{subfigure}[b]{0.45\textwidth}
\centering
\includegraphics[width=\textwidth]{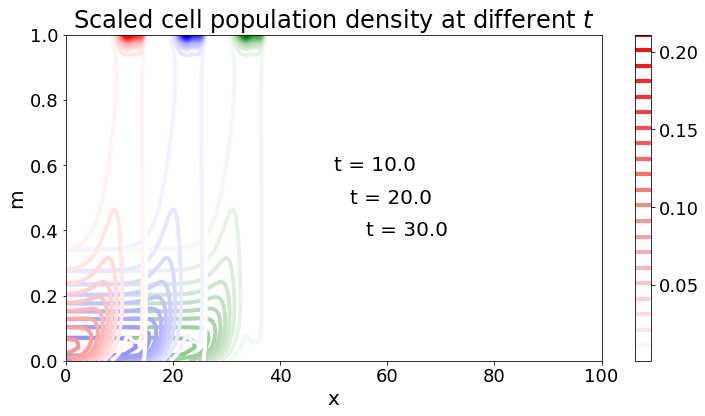}
\caption{\textbf{IIIB3}}
\label{fig:IIIB3_mx}
\end{subfigure}

\caption{(a): Total cell density (red lines) and concentration of chemoattractant (black dotted line) in space, at different time points for different scenarios; (b): Scaled cell population density $p(x,t,m)/P(x,t)$ at different time points in the $(x,m)$ domain; Parameters and functions used: $D(m) \equiv 1$, $b_0=1.5$, $\beta=0.75$, $K_1 = 0.75$, $K_2 = 0.1$, $d=1$, $q=0.5$, $\sigma =1$, $\omega_1(m) := \frac{1-m}{0.1+1-m}$, $\omega_2(m) := \frac{m}{0.1+m}$, $\epsilon = 0.1$, $B(m,y):=B_1(m,y)$, with $B_1$ defined via (5.4), $A(m):=2m$, $b(m) :=b_0(1-m)c$. }
        \label{fig:IIIB3_SI}
\end{figure}

\end{document}